\documentclass{article}

\PassOptionsToPackage{numbers, compress}{natbib}

\usepackage[final]{neurips_2021}

\usepackage[utf8]{inputenc} %
\usepackage[T1]{fontenc}    %
\usepackage{hyperref}       %
\usepackage{url}     
\usepackage{bbm}
\usepackage{longtable}
\usepackage{booktabs} 
\usepackage{algorithm}
\usepackage[noend]{algorithmic} %
\usepackage{tikz}
\usepackage{mathtools} %
\usepackage{amsfonts,amsthm}
\usepackage{thm-restate}
\usepackage{thmtools}      %
\usepackage{nicefrac}       %
\usepackage{microtype}      %
\usepackage{xcolor}         %

\newtheorem{definition}{Definition}
\newtheorem{proposition}{Proposition}

\newcommand{\dtv}{\ensuremath{d_{TV}}}

\newcommand{\primeimp}{$\mathsf{PI}$}

\newcommand{\dnf}{$\mathsf{DNF}$}
\newcommand{\cnf}{$\mathsf{CNF}$}
\newcommand{\nnf}{$\mathsf{NNF}$}

\newcommand{\ddnnf}{$\mathsf{d}$-$\mathsf{DNNF}$}
\newcommand{\dnnf}{$\mathsf{DNNF}$}
\newcommand{\sdnnf}{$\mathsf{SDNNF}$}

\newcommand{\con}{$\mathsf{CON}$}
\newcommand{\mar}{$\mathsf{MAR}$}

\newcommand{\accept}{$\mathsf{Accept}$}
\newcommand{\reject}{$\mathsf{Reject}$}
 
\newcommand{\waps}{$\mathsf{WAPS}$}
\newcommand{\teq}{\mathsf{Teq}}
\newcommand{\peq}{\mathsf{Peq}}
\newcommand{\asamp}{\mathsf{Samp}}
\newcommand{\awct}{\mathsf{Awct}}

\newcommand{\alglinelabel}{%
	\addtocounter{ALC@line}{-1}%
	\refstepcounter{ALC@line}%
	\label%
}

\title{Testing Probabilistic Circuits \thanks{The accompanying tool, available open source, can be found at \href{https://github.com/meelgroup/barbarik}{https://github.com/meelgroup/teq}. The Appendix is available in the accompanying supplementary material.} \thanks{The authors decided to forgo the old convention of alphabetical ordering of authors in favor of a randomized ordering, denoted by \textcircled{r}. The publicly verifiable record of the randomization is available at \protect\url{https://www.aeaweb.org/journals/policies/random-author-order/search} with confirmation code: icoxrj0sNB2L. For citation of the work, authors request that the citation guidelines by AEA for random author ordering be followed. }}

\author{%
	Yash Pote \quad \textcircled{r} \quad    Kuldeep S. Meel \\
	School of Computing, National University of Singapore
}

\begin{document}

\maketitle

\begin{abstract}
Probabilistic circuits (PCs) are a powerful modeling framework for representing tractable probability distributions over combinatorial spaces. In machine learning and probabilistic programming, one is often interested in understanding whether the distributions learned using PCs are close to the desired distribution. Thus, given two probabilistic circuits, a fundamental problem of interest is to determine whether their distributions are close to each other.

The primary contribution of this paper is a closeness test for PCs with respect to the total variation distance metric. Our algorithm utilizes two common PC queries, counting and sampling. In particular, we provide a poly-time probabilistic algorithm to check the closeness of two PCs, when the PCs support tractable approximate counting and sampling. We demonstrate the practical efficiency of our algorithmic framework via a detailed experimental evaluation of a prototype implementation against a set of 475 PC benchmarks. We find that our test correctly decides the closeness of all 475 PCs within 3600 seconds. 
\end{abstract}
\section{Introduction}\label{introduction}

Probabilistic modeling is at the heart of modern computer science, with applications ranging from image recognition and image generation~\cite{PL00,RMC15} to weather forecasting~\cite{CSG04}. Probabilistic models have a multitude of representations, such as probabilistic circuits (PCs)~\cite{CVV20}, graphical models~\cite{KF09}, generative networks~\cite{GPMX+}, and determinantal point processes~\cite{K12}. Of particular interest to us are PCs, which are known to support guaranteed inference and thus have applications in safety-critical fields such as healthcare~\cite{AH98,ODW00}. In this work, we will focus on PCs that are fragments of the Negation Normal Form ({\nnf}), specifically {\dnnf}s, {\ddnnf}s, {\sdnnf}s, and {\primeimp}s~\cite{DH02}. We refer to the survey by~\citet{CVV20} for more details regarding PCs.

Given two distributions $P$ and $Q$, a fundamental problem is to determine whether they are close. Closeness between distributions is frequently quantified using the total variation (TV) distance, $\dtv(P,Q)=(1/2)\| P-Q \|_1,$ where $\|\cdot\|$ is the $\ell_1$ norm~\cite{LKFO18,CDKS20}. Thus, stated formally, closeness testing is the problem of deciding whether $\dtv(P,Q)\leq \varepsilon$ or $\dtv(P,Q)\geq \eta$ for $0 \leq \varepsilon<\eta \leq 1$.  Determining the closeness of models has applications in AI planning~\cite{DH02}, bioinformatics~\cite{RESG14,SM15,YGMG15} and probabilistic program verification~\cite{DLHM18,MO05}.

Equivalence testing is a special case of closeness testing, where one tests if $\dtv(P,Q)= 0$.  \citet{DH02} initiated the study of equivalence testing of PCs by designing an equivalence test for {\ddnnf}s. An equivalence test is, however, of little use in contexts where the PCs under test encode non-identical distributions that are nonetheless close enough for practical purposes. Such situations may arise due to the use of approximate PC compilation~\cite{CT20} and sampling-based learning of PCs~\cite{PLVM20,PVSM+20}. As a concrete example, consider PCs that are learned via approximate methods such as stochastic gradient descent~\cite{PVSM+20}. In such a case, PCs are likely to converge to close but non-identical distributions. Given two such PCs, we would like to know whether they have converged to distributions close to each other. Thus, we raise the question: \textit{Does there exist an efficient algorithm to test the closeness of two PC distributions?}

In this work, we design the first closeness test for PCs with respect to TV distance, called $\teq$. Assuming the tested PCs allow poly-time approximate weighted model counting and sampling, $\teq$ runs in polynomial time. Formally, given two PC distributions $P$ and $Q$, and three parameters ($\varepsilon$,$\eta$,$\delta$), for closeness($\varepsilon$), farness($\eta$), and  tolerance($\delta$), $\teq$ returns {\accept} if $\dtv(P,Q) \leq \varepsilon$ and {\reject} if $\dtv(P,Q)\geq \eta$ with probability at least $1-\delta$.  $\teq$ makes atmost $O((\eta-\varepsilon)^{-2}\log(\delta^{-1}))$ calls to the sampler and exactly 2 calls to the counter. 

$\teq$ builds on a general distance estimation technique of~\citet{CR14} that estimates the distance between two distributions with a small number of samples. In the context of PCs, the algorithm requires access to an exact sampler and an exact counter. Since not all PCs support exact sampling and counting, we modify the technique presented in~\cite{CR14} to allow for approximate samples and counts. Furthermore, we implement and test $\teq$ on a dataset of publicly available PCs arising from applications in circuit testing. Our results show that closeness testing can be accurate and scalable in practice. 

For some {\nnf} fragments, such as {\dnnf}, no sampling algorithm is known, and for fragments such as {\primeimp}, sampling is known to be NP-hard~\cite{R96}. Since $\teq$ requires access to approximate weighted counters and samplers to achieve tractability, the question of determining the closeness of the PCs mentioned above remains unanswered. Thus, we investigate further and characterize the complexity of closeness testing for a broad range of PCs. Our characterization reveals that PCs from the fragments {\ddnnf}s and {\sdnnf}s can be tested for closeness in poly-time via $\teq$, owing to the algorithms of~\citet{D01a} and~\citet{ACJR20}. We show that the {\sdnnf} approximate counting algorithm of~\citet{ACJR20} can be extended to log-linear {\sdnnf}s using chain formulas~\cite{CFMV15}. Then, using previously known results, we also find that there are no poly-time equivalence tests for PCs from {\primeimp} and {\dnnf}, conditional on widely believed complexity-theoretic conjectures. Our characterization also reveals some open questions regarding the complexity of closeness and equivalence testing of PCs.

The rest of the paper is organized in the following way. We define the notation and discuss related work in Section~\ref{preliminaries}. We then present the main contribution of the paper, the closeness test $\teq$, and the associated proof of correctness in Section~\ref{theory}. We present our experimental findings in Section~\ref{evaluation}, and then discuss the complexity landscape of closeness testing in Section~\ref{sec:characterization}. We conclude the paper and discuss some open problems in Section~\ref{conclusion}. Due to space constraints, we defer some proofs to the supplementary Section~\ref{appendix}.
\section{Background}\label{preliminaries}
Let $\varphi: \{0,1\}^n \rightarrow \{0,1\}$ be a circuit over $n$ Boolean variables. An assignment $\sigma \in \{0,1\}^{n}$ to the variables of $\varphi$ is a \textit{satisfying assignment} if $\varphi(\sigma) = 1$. The set of all satisfying assignments of $\varphi$ is $R_\varphi$. If $|R_{\varphi}|>0$, then $\varphi$ is said to be \emph{satisfiable} and if $|R_{\varphi}| = 2^n$, then $\varphi$ is said to be \emph{valid}. We use $|\varphi|$ to denote the size of circuit $\varphi$, where the size is the total number of vertices and edges in the circuit DAG.

The polynomial hierarchy (PH) contains the classes $\Sigma_1^\text{P}$ (NP) and $\Pi_1^\text{P}$ (co-NP) along with generalizations of the form $\Sigma_{i}^\text{P}$ and $\Pi_{i}^\text{P}$ where $\Pi^\text{P}_{i+1} = \text{co-NP}^{\Pi_i^P}$ and $\Sigma_{i+1}^P = \text{NP}^{\Sigma_i^P}$~\cite{S76}. The classes ${\Sigma_i^P}$ and ${\Pi_i^P}$ are said to be at level $i$. If it is shown that two classes on the same or consecutive levels are equal, the hierarchy collapses to that level. Such a collapse is considered unlikely, and hence is used as the basic assumption for showing hardness results, including the ones we present in the paper.

\subsection{Probability distributions}

A weight function $\mathtt{w}: \{0,1\}^n \rightarrow \mathbb{Q}^+$ assigns a positive rational weight to each assignment $\sigma$. We extend the definition of $\mathtt{w}$ to also allow circuits as input: $\mathtt{w}(\varphi) =  \underset{\sigma \in R_{\varphi}}{\sum}\mathtt{w}(\sigma)$.
For weight function $\mathtt{w}$ and circuit $\varphi$,  $\mathtt{w}(\varphi)$ is  the weighted model count (WMC) of $\varphi$ w.r.t. $\mathtt{w}$. 

In this paper, we focus on log-linear weight functions as they capture a wide class of distributions, including those arising from graphical models, conditional random fields, and skip-gram models~\cite{M12}. Log-linear models are represented as literal-weighted functions, defined as:
\begin{definition}
	For a set $X$ of $n$ variables,  a weight function $\mathtt{w}$ is called literal-weighted if there is a poly-time computable map $\mathtt{w}:X \rightarrow \mathbb{Q} \cap (0,1)$ such that for any assignment $\sigma \in \{0,1\}^n:$
	\begin{align*}
		\mathtt{w}(\sigma) = \prod_{x \in \sigma} 
		\begin{cases}
			\mathtt{w}(x)&if \quad x=1 \\
			1-\mathtt{w}(x)&if \quad x=0 
		\end{cases}
	\end{align*}
\end{definition}
For all circuits $\varphi$, and log-linear weight functions $\mathtt{w}$,  $\mathtt{w}(\varphi)$ can be represented in size polynomial in the input.

\paragraph{Probabilistic circuits:} A probabilistic circuit is a satisfiable circuit $\varphi$ along with a weight function $\mathtt{w}$. $\varphi$ and $\mathtt{w}$ together define a discrete probability distribution on the set $\{0,1\}^n$ that is supported over $R_\varphi$. We denote the p.m.f. of this distribution as: 
$
P(\varphi, \mathtt{w})(\sigma) = 
\begin{cases}
	0& \varphi(\sigma) = 0\\
	\mathtt{w}(\sigma)/\mathtt{w}(\varphi)& \varphi(\sigma) = 1
\end{cases}$

In this paper, we study circuits that are fragments of the Negation Normal Form ({\nnf}). A circuit $\varphi$ in {\nnf} is a rooted, directed acyclic graph (DAG), where each leaf node is labeled with true, false, $v$ or $\lnot v$; and each internal node is labeled with a $\wedge$ or $\vee$ and can have arbitrarily many children. We focus on four fragments of {\nnf}, namely, Decomposable {\nnf}(\dnnf), deterministic-{\dnnf}(\ddnnf), Structured {\dnnf}(\sdnnf), and Prime Implicates(\primeimp). For further information regarding circuits in {\nnf}, refer to the survey~\cite{D02} and the paper~\cite{PD08}.

The TV distance  of two probability distributions $P(\varphi_1,\mathtt{w_1})$ and $P(\varphi_2,\mathtt{w_2})$ over $\{0,1\}^n$ is defined as:  
$\dtv(P(\varphi_1,\mathtt{w_1}),P(\varphi_2,\mathtt{w_2}))= \frac{1}{2}\sum_{\sigma \in \{0,1\}^n} |P(\varphi_1,\mathtt{w_1})(\sigma) - P(\varphi_2,\mathtt{w_2})(\sigma)|$.

$P(\varphi_1,\mathtt{w_1})$ and $P(\varphi_2,\mathtt{w_2})$ are said to be (1) equivalent if $\dtv(P(\varphi_1,\mathtt{w_1}), P(\varphi_2,\mathtt{w_2})) = 0$, (2) $\varepsilon$-close if   $\dtv(P(\varphi_1,\mathtt{w_1}), P(\varphi_2,\mathtt{w_2})) \leq \varepsilon$, and (3) $\eta$-far if $\dtv(P(\varphi_1,\mathtt{w_1}), P(\varphi_2,\mathtt{w_2}))\geq  \eta$.

Our closeness testing algorithm $\teq$, assumes access to an approximate weighted counter $\awct(\alpha, \beta, \varphi , \mathtt{w})$, and an approximate weighted sampler $\asamp(\alpha, \beta, \varphi , \mathtt{w})$. We define their behavior as follows:

\begin{definition}\label{def:awct}
	$\awct(\alpha, \beta, \varphi , \mathtt{w})$ takes a circuit $\varphi$, a weight function $\mathtt{w}$, a tolerance parameter $\alpha >  0$ and a confidence parameter $\beta > 0$ as input and returns the approximate weighted model count of $\varphi$ w.r.t. $\mathtt{w}$ such that
	\begin{align*}
		\Pr\left[\frac{\mathtt{w}(\varphi)}{1+\alpha} \leq {\awct}(\alpha, \beta, \varphi , \mathtt{w}) \leq (1+\alpha)\mathtt{w}(\varphi)  \right] \geq 1-\beta
	\end{align*}
	Tractable approximate counting algorithms for PCs are known as Fully Polynomial Randomised Approximation Schemes (FPRAS). The running time of an FPRAS is given by $T(\alpha, \beta, \varphi) = poly(\alpha^{-1}, \log(\beta^{-1}),|\varphi|)$.
\end{definition}

\begin{definition}\label{def:asamp}
	$\asamp(\alpha, \beta, \varphi , \mathtt{w})$ takes a circuit $\varphi$, a weight function $\mathtt{w}$, a tolerance parameter $\alpha > 0$ and a confidence parameter $\beta > 0$ as input and returns either (1) a satisfying assignment $\sigma$ sampled approximately w.r.t. weight function $\mathtt{w}$ with probability $\geq  1-\beta $ or (2) a symbol $\perp$ indicating failure with probability  $<\beta$. In other words, whenever $\asamp$ samples $\sigma$: 
	\begin{align*}
		\frac{P(\varphi,\mathtt{w})(\sigma)}{1+\alpha} \leq \Pr[\asamp(\alpha, \beta, \varphi , \mathtt{w}) = \sigma] \leq (1+\alpha)P(\varphi,\mathtt{w})(\sigma)
	\end{align*}
	Tractable approximate sampling algorithms for PCs are known as Fully Polynomial Almost Uniform Samplers (FPAUS). The running time of an FPAUS for a single sample is given by $T(\alpha, \beta, \varphi) = poly(\alpha^{-1}, \log(\beta^{-1}),|\varphi|)$.
\end{definition}

In the rest of the paper $[m]$ denotes the set $\{1,2,\ldots m\}$, $\mathbbm{1}(e)$ represents the indicator variable for event $e$, and $\mathbbm{E}(v)$ represents the expectation of random variable $v$. 

\subsection{Related work}

\paragraph{Closeness testing:} 
Viewing circuit equivalence testing through the lens of distribution testing, we see that the {\ddnnf} equivalence test of~\citet{DH02} can be interpreted as an equivalence test for uniform distribution on the satisfying assignments of {\ddnnf}s. This relationship between circuit equivalence testing and closeness testing lets us rule out the existence of distributional equivalence tests for all those circuits for which circuit equivalence is already known to be hard under complexity-theoretic assumptions. We will explore this further in Section~\ref{sec:hardness}.

\paragraph{Distribution testing:} Discrete probability distributions are typically defined over an exponentially large number of points; hence a lot of recent algorithms research has focused on devising tests that require access to only a sublinear or even constant number of points in the distribution~\cite{C20}. In this work, we work with distributions over $\{0,1\}^n$, and thus we aim to devise algorithms with running time at most polynomial in $n$. Previous work in testing distributions over Boolean functions has focused on the setting where the distributions offer pair-conditional sampling access~\cite{CM19,MPC20}. Using pair-conditional sampling access,~\citet{MPC20} were able to test distributions for closeness using $\Tilde{O}(tilt(\varphi)^2/(\eta-6\varepsilon)^2\eta)$ queries, where $tilt$ is the ratio of the probabilities of the most and least probable element in the support.

\section{$\teq$: a tractable algorithm for closeness testing}\label{theory}

In this section, we present the main contribution of the paper: a closeness test for PCs, $\teq$. The pseudocode of $\teq$ is given in Algorithm~\ref{alg:epsapproximateprobeq}. 

Given satisfiable circuits $\varphi_1,\varphi_2$ and weight functions $\mathtt{w_1},\mathtt{w_2}$ along with parameters $(\varepsilon,\eta,\delta)$, $\teq$ decides whether the TV distance between $P(\varphi_1, \mathtt{w_1})$ and $P(\varphi_2, \mathtt{w_2})$ is lesser than $\varepsilon$ or greater than $\eta$ with confidence at least $1-\delta$. $\teq$ assumes access to an approximate weighted counter $\awct(\alpha, \beta, \varphi , \mathtt{w})$, and an approximate weighted sampler $\asamp(\alpha, \beta, \varphi , \mathtt{w})$. We define their behavior in the following two definitions.

\paragraph{The algorithm}
$\teq$ starts by computing constants $\gamma$ and $m$. Then it queries the $\awct$ routine with circuit $\varphi_1$ and weight function $\mathtt{w_1}$ to obtain a $\sqrt{1+\gamma/4}-1$ approximation of $\mathtt{w_1}(\varphi_1)$ with confidence at least $1-\delta/8$. A similar query is made for $\varphi_2$ and $\mathtt{w_2}$ to obtain an approximate value for $\mathtt{w_2}(\varphi_2)$. These values are stored in $k_1$ and $k_2$, respectively. $\teq$ maintains a $m$-sized array $\Gamma$, to store the estimates for $r(\sigma_i)$. $\teq$ now iterates $m$ times. In each iteration, it generates one sample $\sigma_i$ through the $\asamp$ call on line~\ref{line:asamp}. There is a small probability of at most $\delta/4m$  that this call fails and returns $\perp$. $\teq$ only samples from one of the two PCs.

The algorithm then proceeds to compute the weight of assignment $\sigma_i$ w.r.t. the weight functions $\mathtt{w_1}$ and $\mathtt{w_2}$ and stores it in $s_1$ and $s_2$, respectively. Using the weights and approximate weighted counts stored in $k_1, k_2$ the algorithm computes the value $r(\sigma_i)$ on line~\ref{line:r2}, where $r(\sigma_i)$ is an approximation of the ratio of the probability of $\sigma_i$ in the distribution $P(\varphi_2, \mathtt{w_2})$ to its probability in $P(\varphi_1, \mathtt{w_1})$. Since $\sigma_i$ was sampled from $P(\varphi_1, \mathtt{w_1})$, its probability in $P(\varphi_1, \mathtt{w_1})$ cannot be 0,  ensuring that there is no division by 0. If the ratio $r(\sigma_i)$ is less than 1, then $\Gamma[i]$ is updated with the value $1-r(\sigma_i)$ otherwise the value of $\Gamma[i]$ remains 0. After the $m$ iterations, $\teq$ sums up the values in the array $\Gamma$. If the sum is found to be less than threshold $m(\varepsilon + \gamma)$, $\teq$ returns {\accept} and otherwise returns {\reject}.

\begin{algorithm}[]
	\caption{$\teq(\varphi_1,\mathtt{w_1}, \varphi_2,\mathtt{w_2},\varepsilon,\eta,  \delta)$}\label{alg:epsapproximateprobeq}
	\begin{algorithmic}[1]
		\STATE $\gamma \gets (\eta-\varepsilon)/2$
		\STATE $m \gets \lceil 2\log(4/\delta)/\gamma^2\rceil $ \alglinelabel{line:m2}
		\STATE	$\Gamma \gets [0]*m$
		\STATE	$k_1 \gets \awct(\sqrt{1+\gamma/4}-1, \delta/8, \varphi_1,\mathtt{w_1})$\alglinelabel{line:awct1}
		\STATE	$k_2 \gets \awct(\sqrt{1+\gamma/4}-1, \delta/8, \varphi_2,\mathtt{w_2})$\alglinelabel{line:awct2}
		\FOR{$i \in \{1,2\ldots,m\}$}
		\STATE$ \sigma_i \gets \asamp(\gamma/(4\eta-2 \gamma),\delta/4m,\varphi_1,\mathtt{w_1})$\alglinelabel{line:asamp}\\
		\IF{$\sigma_i \not = \perp$}
		\STATE $s_1 \gets \mathtt{w_1}(\sigma_i),s_2 \gets \mathtt{w_2}(\sigma_i)$
		\STATE $r(\sigma_i) \gets \frac{s_2}{k_2}\cdot \frac{k_1}{s_1}$\alglinelabel{line:r2}
		\IF{$r(\sigma_i) <1$\alglinelabel{line:checkri}}
		\STATE $\Gamma[i] \gets 1-r(\sigma_i)$\alglinelabel{line:setgammavalue}
		\ENDIF
		\ENDIF
		\ENDFOR
		\IF{$\sum_{i \in [m]} \Gamma[i] \leq m(\varepsilon+\gamma)$\alglinelabel{line:if2}}
		\STATE {\textbf{Return} {\accept}}
		\ELSE
		\STATE { \textbf{Return} {\reject}}
		\ENDIF
	\end{algorithmic}
\end{algorithm}
The following theorem asserts the correctness of $\teq$.  
\begin{restatable}{thm}{tolapprox}\label{thm:tolapprox}
	Given two satisfiable probabilistic circuits $\varphi_1,\varphi_2$ and weight functions $\mathtt{w_1},\mathtt{w_2}$, along with parameters $\varepsilon<\eta<1$ and $\delta <1$,
	\begin{enumerate}
		\item[A.] If $\dtv(P(\varphi_1,\mathtt{w_1}), P(\varphi_2,\mathtt{w_2})) \leq \varepsilon$, then $\teq(\varphi_1,\mathtt{w_1}, \varphi_2,\mathtt{w_2},\varepsilon,\eta,  \delta)$ returns {\accept} with probability at least $ (1-\delta)$.
		\item[B.] If $\dtv(P(\varphi_1,\mathtt{w_1}),P(\varphi_2,\mathtt{w_2})) \geq  \eta$, then $\teq(\varphi_1,\mathtt{w_1}, \varphi_2,\mathtt{w_2},\varepsilon,\eta,  \delta)$ returns {\reject} with probability at least $ (1-\delta)$.
	\end{enumerate} 
\end{restatable}
The following theorem states the running time of the algorithm,
\begin{restatable}{thm}{runtime}\label{thm:runtime}
Let $\gamma = \eta -\varepsilon$, then the time complexity of $\teq$ is in $O\left(T_{\awct}(\gamma,\delta, \max(|\varphi_1|,|\varphi_2|)) + T_{\asamp}(\gamma,\delta, \max(|\varphi_1|,|\varphi_2|))\frac{\log(\delta^{-1})}{\gamma^2}\right)$. If the underlying PCs support approximate counting and sampling in polynomial time, then the running time of $\teq$ is also polynomial in terms of $\gamma, \log(\delta^{-1})$ and $max(|\varphi_1|, |\varphi_2|)$.  
\end{restatable}

To improve readability, we use $P_1$ to refer to the distribution $P(\varphi_1,\mathtt{w_1})$ and $P_2$ to refer to  $P(\varphi_2,\mathtt{w_2})$.

\subsection{Proving the correctness of $\teq$}

In this subsection, we present the theoretical analysis of $\teq$, and the proof of  Theorem~\ref{thm:tolapprox}(A). We will defer the proofs of Theorem~\ref{thm:tolapprox}(B) and Theorem~\ref{thm:runtime} to the supplementary Section~\ref{sec:proof6} and Section~\ref{sec:proof7}, respectively.

For the purpose of the proof, we will first define events $\mathtt{Pass_1}, \mathtt{Pass_2}$ and $\mathtt{Good}$. Events $\mathtt{Pass_1}, \mathtt{Pass_2}$ are defined w.r.t. the function calls $\awct(\sqrt{1+\gamma/4}-1, \delta/8, \varphi_1,\mathtt{w_1})$ and $\awct(\sqrt{1+\gamma/4}-1, \delta/8, \varphi_2,\mathtt{w_2})$, respectively (as on lines~\ref{line:awct1},~\ref{line:awct2}  of Algorithm~\ref{alg:epsapproximateprobeq}). $\mathtt{Pass_1}$ and $\mathtt{Pass_2}$ represent the events that the two calls correctly return $\sqrt{1+\gamma/4}$  approximations of the weighted model counts of $\varphi_1$ and $\varphi_2$ i.e. $
	\frac{\mathtt{w_1}(\varphi_1)}{\sqrt{1+\gamma/4}}\leq \awct(\sqrt{1+\gamma/4}-1, \delta/8,\varphi_1, \mathtt{w_1}) \leq  (\sqrt{1+\gamma/4})\mathtt{w_1}(\varphi_1)$, and $
	\frac{\mathtt{w_2}(\varphi_2)}{\sqrt{1+\gamma/4}}\leq \awct(\sqrt{1+\gamma/4}-1, \delta/8,\varphi_2, \mathtt{w_2}) \leq  (\sqrt{1+\gamma/4} )\mathtt{w_2}(\varphi_2)
$. From the definition of $\awct$,  we have $\Pr[\mathtt{Pass_1}],\Pr[\mathtt{Pass_2}] \geq 1-\delta/8$. 

Let $\mathtt{Fail_i}$ denote the event that $\asamp$ (Algorithm~\ref{alg:epsapproximateprobeq}, line~\ref{line:asamp}) returns the symbol $\perp$ in the $i$th iteration of the loop. By the definition of $\asamp$ we know that $\forall_{i \in [m]}\Pr[\mathtt{Fail_i}]< \delta/4m$.

The analysis of $\teq$ requires that all $m$ $\asamp$ calls and both $\awct$ calls return correctly. We denote this super-event as $\mathtt{Good} = \bigcap_{i \in [m]}\overline{\mathtt{Fail_i}} \cap \mathtt{Pass_1} \cap \mathtt{Pass_2}$. Applying the union bound we see that the probability of all calls to $\awct$ and $\asamp$ returning without error is at least $1-\delta/2$:
\begin{align}
	\Pr[\mathtt{Good}]  &= 1- \Pr[\bigcup_{i \in [m]} \mathtt{Fail_i} \cup \overline{\mathtt{Pass_1}} \cup \overline{\mathtt{Pass_2}}]\geq 1 - m\cdot \delta/4m -2\cdot \delta/8 = 1- \delta/2 \label{line:excludebadevents} 
\end{align}

We will now state a lemma, which we will prove in the supplementary Section~\ref{sec:proof9}.

\begin{restatable}{lem}{allerrors}\label{lem:allerrors}
	$\mathtt{Good} \rightarrow  \left| r(\sigma)-\frac{P_2(\sigma)}{P_1(\sigma)} \right| \leq  \gamma/4\cdot \frac{P_2(\sigma)}{P_1(\sigma)}  $
\end{restatable}

We now prove the lemma critical for our proof of correctness of $\teq$.
\begin{restatable}{lem}{boundonA}\label{lem:boundonA}
	Assuming the event $\mathtt{Good}$, let $A = \sum_{\sigma \in \{0,1\}^n} \mathbbm{1}\left(r(\sigma)<1\right)\left(1 -r(\sigma)\right) P_1(\sigma)
	$ , then
	\begin{enumerate}
		\item[1.] If 	$\dtv(P_1,P_2) \leq \varepsilon$	, then $A \leq  \varepsilon +\gamma/4 $
		\item[2.] If    $\dtv(P_1,P_2)	\geq \eta$ , then $A \geq   \eta - \gamma/4$
	\end{enumerate} 
\end{restatable}

\begin{proof}
 If \;$ \sum_{x } \left(P_1(x)- P_2(x)\right) = 0$, then $\frac{1}{2} \sum_{x}|P_1(x)- P_2(x)| = \underset{x: P_1(x)- P_2(x) > 0}{\sum} (P_1(x)- P_2(x))$. Using this fact we see that,
	\begin{align*}
		\dtv(P_1,P_2)  
		&=\sum_{\sigma:P_2(\sigma)< P_1(\sigma) } P_1(\sigma) - P_2(\sigma) =\sum_{\sigma:\frac{P_2(\sigma)}{ P_1(\sigma)}<1} \left(1 - \frac{P_2(\sigma)}{P_1(\sigma)}\right) P_1(\sigma)   \\	
		&=\sum_{\sigma \in \{0,1\}^n} \mathbbm{1}\left(\frac{P_2(\sigma)}{ P_1(\sigma)}<1\right) \left(1 - \frac{P_2(\sigma)}{P_1(\sigma)}\right)  P_1(\sigma) \nonumber \\
		&=A +   \underbrace{\sum_{\sigma \in \{0,1\}^n}  \mathbbm{1}\left(\frac{P_2(\sigma)}{ P_1(\sigma)}<1\right) \left(1 - \frac{P_2(\sigma)}{P_1(\sigma)}\right)P_1(\sigma)- A}_{B}    
	\end{align*}
	Thus we have that $\dtv(P_1,P_2)	- A = B$.	We now divide the set of assignments $\sigma \in \{0,1\}^n$ into three disjoint partition $S_1, S_2$ and $S_3$ as following: $S_1= \{ \sigma : \mathbbm{1}(\frac{P_2(\sigma)}{P_1(\sigma)} <1) =  \mathbbm{1}(r(\sigma)<1) \}$; $S_2= \{ \sigma : \mathbbm{1}(\frac{P_2(\sigma)}{P_1(\sigma)} <1) >  \mathbbm{1}(r(\sigma)<1) \}$;
	 $S_3= \{ \sigma : \mathbbm{1}(\frac{P_2(\sigma)}{P_1(\sigma)} <1) <  \mathbbm{1}(r(\sigma)<1) \}$. The definition implies that the indicator $\mathbbm{1}(r(\sigma)<1)$ is $0$ for all assignments in the set $S_2$, and is $1$ for all assignments in $S_3$. Similarly   $\mathbbm{1}(\frac{P_2(\sigma)}{P_1(\sigma)} <1)$ takes value $1$ and $0$ for all elements in $S_2$ and $S_3$, respectively.
	
	Now we bound the magnitude of $B$,
	\begin{align*}
		&|B| =\left| \sum_{\sigma \in \{0,1\}^n} \left[  \left(1 - \frac{P_2(\sigma)}{P_1(\sigma)}\right) \mathbbm{1}\left(\frac{P_2(\sigma)}{ P_1(\sigma)}<1\right) - \left(1 - r(\sigma)\right) 
		\mathbbm{1}\left(r(\sigma)<1\right)  \right] P_1(\sigma)\right|
	\end{align*}
	For $b_j>0$, we have that  $|\sum_j a_j b_j| \leq \sum_j |a_j|b_j$, and thus:
	\begin{align*}
|B| \leq
		  \sum_{\sigma \in \{0,1\}^n} \left| \left[  \left(1 - \frac{P_2(\sigma)}{P_1(\sigma)}\right)  \mathbbm{1}\left(\frac{P_2(\sigma)}{ P_1(\sigma)}<1\right) - \left(1 - r(\sigma)\right) 
		\mathbbm{1}\left(r(\sigma)<1\right)  \right] \right| P_1(\sigma)
		\end{align*}
	We can split the summation into three terms based on the sets in which the assignments lie. Some summands take the value $0$ in a particular set, so we don't include them in the term.
	\begin{align*}
		|B|\leq&
		  \sum_{\sigma \in S_1} \mathbbm{1}\left(\frac{P_2(\sigma)}{ P_1(\sigma)}<1\right)\left|r(\sigma) - \frac{P_2(\sigma)}{ P_1(\sigma)}\right|P_1(\sigma) + \sum_{\sigma \in S_2}\mathbbm{1}\left(\frac{P_2(\sigma)}{ P_1(\sigma)}<1\right)\left(1 - \frac{P_2(\sigma)}{P_1(\sigma)}\right)P_1(\sigma)\\
		&+ \sum_{\sigma \in S_3} \mathbbm{1}\left(r(\sigma)<1\right)\left(1-r(\sigma)\right)P_1(\sigma)
	\end{align*}	
	Since we know that  $\forall \sigma \in S_2, r(\sigma)>1$ and $\forall  \sigma \in S_3, \frac{P_2(\sigma)}{ P_1(\sigma)}>1$, we can alter the second and third terms of the inequality in the following way:
	\begin{align*}
		|B|\leq& \sum_{\sigma \in S_1} \mathbbm{1}\left(\frac{P_2(\sigma)}{ P_1(\sigma)}<1\right)\left|r(\sigma) - \frac{P_2(\sigma)}{ P_1(\sigma)}\right| P_1(\sigma)+ \sum_{\sigma \in S_2}\mathbbm{1}\left(\frac{P_2(\sigma)}{ P_1(\sigma)}<1\right)\left|r(\sigma) - \frac{P_2(\sigma)}{ P_1(\sigma)}\right| P_1(\sigma)\\
		&+ \sum_{\sigma \in S_3} \mathbbm{1}\left(r(\sigma)<1\right)\left| \frac{P_2(\sigma)}{ P_1(\sigma) }-r(\sigma)\right|P_1(\sigma) \leq \sum_{\sigma \in S_1 \cup S_2 \cup S_3} \left|r(\sigma) - \frac{P_2(\sigma)}{ P_1(\sigma)}\right| P_1(\sigma)
	\end{align*}
	Using our assumption of the event $\mathtt{Good}$ and Lemma~\ref{lem:allerrors},	$|B| \leq \sum_{\sigma \in \{0,1\}^n}\gamma/4 \cdot P_1(\sigma) \leq \gamma/4$
	Since $\dtv(P_1,P_2)	- A = B$, we get
		$|\dtv(P_1,P_2)	- A| \leq \gamma/4$. We can now deduce that if 
	$
	\dtv(P_1,P_2) \leq \varepsilon
	$, then $A \leq  \varepsilon + \gamma/4 $
	and if $
	\dtv(P_1,P_2) \geq \eta
	$, then $		A \geq  \eta - \gamma/4$.
\end{proof}	
\paragraph{Using $\teq$ to test PCs in general.}
Exact weighted model counting(WMC) is a commonly supported query on PCs. In the language of PC queries, a WMC query is known as the marginal (\mar) query. Conditional inference (\con) is another well studied PC query. Using {\con} and {\mar}, one can sample from the distribution encoded by a given PC. It is known that if a PC has the structural properties of \textit{smoothness} and \textit{decomposability}, then the {\con} and {\mar} queries can be computed tractably. For the definitions of the above terms and further details, please refer to the survey~\cite{CVV20}.  
\section{Evaluation}\label{evaluation}

To evaluate the performance of $\teq$, we implemented a prototype in Python. The prototype uses \waps\footnote{\hyperlink{https://github.com/meelgroup/WAPS}{https://github.com/meelgroup/WAPS}}~\cite{GSRM19} as a weighted sampler to sample over the input {\ddnnf} circuits. The primary objective of our experimental evaluation was to seek an answer to the following question: Is $\teq$ able to determine the closeness of a pair of probabilistic circuits by returning {\accept} if the circuits are $\varepsilon$-close and {\reject} if they are $\eta$-far? We test our tool ${\teq}$ in the following two settings: 
\begin{enumerate}
	\item[\textbf{A}.] The pair of PCs represent small randomly generated circuits and weight functions.
	\item[\textbf{B}.] The pair of PCs are from the set of publicly available benchmarks arising from sampling and counting tasks. 
\end{enumerate} 

Our experiments were conducted on a high performance compute cluster with Intel Xeon(R) E5-2690 v3@2.60GHz CPU cores. For each benchmark, we use a single core with a timeout of 7200 seconds.

\subsection{Setting A - Synthetic benchmarks}
\paragraph{Dataset} Our dataset for experiments conducted in setting \textbf{A} consisted of randomly generated 3-{\cnf}s and with random literal weights. Our dataset consisted of 3-{\cnf}s with $\{14,15,16,17,18\}$ variables. Since the circuits are small, we validate the results by computing the actual total variation distance using brute-force.

\begin{table}[h!]	\centering	\begin{tabular}{cccccc}\toprule & \multicolumn{3}{c}{$\dtv$}&  &    \\ \cmidrule(l){2-4} Benchmark&$\leq\epsilon $&$\geq \eta$&Actual&Result& Expected Result \\ \midrule 
15\_3&0.75&0.94&0.804&R&A/R\\ \midrule 
14\_2&0.8&0.9&0.764&A&A\\ \midrule 
17\_4&0.75&0.9&0.941&R&R\\ \midrule 
14\_1&0.9&0.99&0.740&A&A\\ \midrule 
18\_2&0.75&0.9&0.918&R&R\\  
		\bottomrule\end{tabular}
	 \caption{Runtime performance of $\teq$. We experiment with 375 random PCs with known $\dtv$, and out of the 375 benchmarks we display 5 in the table and the rest in the supplementary Section ~\ref{sec:extended}. In the table `A' represents {\accept} and `R' represents {\reject}. In the last column `A/R' represents that both {\accept} and {\reject} are acceptable outputs for $\teq$. }\label{tab:randcnf_table}\end{table}\raggedbottom

\paragraph{Results}
Our tests terminated with the correct result in less than 10 seconds on all the randomly generated PCs we experimented with.  We present the empirical results in Table~\ref{tab:randcnf_table}. The first column indicates the benchmark's name, the second and third indicate the parameters $\varepsilon$ and $\eta$ on which we executed $\teq$.   The fourth column indicates the actual $\dtv$ distance between the two benchmark PCs. The fifth column indicates the output of $\teq$, and the sixth indicates the expected result. The full detailed results are presented in the appendix Section~\ref{sec:extended}.

\subsection{Setting B - Real-world benchmarks}
\paragraph{Dataset} We conducted experiments on a range of publicly available benchmarks arising from sampling and counting tasks\footnote{\href{https://zenodo.org/record/3793090}{https://zenodo.org/record/3793090}}. Our dataset contained 100 {\ddnnf} circuits with weights. We have assigned random weights to literals wherever weights were not readily available. For the empirical evaluation of $\teq$, we needed pairs of weighted {\ddnnf}s with known $\dtv$ distance. To generate such a dataset, we first chose a circuit and a weight function, and then we synthesized new weight functions using the technique of \textit{one variable perturbation}, described in the appendix Section~\ref{sec:syntheticbenchmark}.

\begin{table}[h!]	\centering	\begin{tabular}{ccccc}\toprule & \multicolumn{2}{c}{$\dtv \leq \varepsilon$}&  \multicolumn{2}{c}{$\dtv \geq \eta$}   \\ \cmidrule(l){2-3}\cmidrule(l){4-5} Benchmark&Result&$\teq$(s)&Result&$\teq$(s)  \\ \midrule or-70-10-8-UC-10&A&23.2&R&22.82\\ 
		s641\_15\_7&A&33.66&R&33.51\\ 
		or-50-5-4&A&414.17&R&408.59\\ 
		ProjectService3&A&356.15&R&356.14\\ 
		s713\_15\_7&A&24.86&R&24.41\\ 
		or-100-10-2-UC-30&A&31.04&R&31.0\\ 
		s1423a\_3\_2&A&153.13&R&152.81\\ 
		s1423a\_7\_4&A&104.93&R&103.51\\ 
		or-50-5-10&A&283.05&R&282.97\\ 
		or-60-20-6-UC-20&A&363.32&R&362.8\\ 
		\bottomrule\end{tabular}
	 \caption{Runtime performance of $\teq$. We experiment with 100 PCs with known $\dtv$, and out of the 100 benchmarks we display 10 in the table and the rest in the appendix ~\ref{sec:extended}. In the table `A' represents {\accept} and `R' represents {\reject}. The value of the closeness parameter is  $\varepsilon = 0.01$ and the farness parameter is $\eta = 0.2$.}\label{tab:comparision}\end{table}\raggedbottom

\paragraph{Results}
We set the closeness parameter $\varepsilon$, farness parameter $\eta$ and confidence $\delta$ for $\teq$ to be $0.01,0.2$ and $0.01$, respectively.  The chosen parameters imply that if the input pair of probabilistic circuits are $\leq 0.01$ close in $\dtv$, then $\teq$ returns {\accept} with probability atleast $0.99$, otherwise if the circuits are $\geq 0.2$ far in $\dtv$, the algorithm returns {\reject} with probability at least $0.99$. The number of samples required for $\teq$ (indicated by the variable $m$ as on line~\ref{line:m2} of Algorithm~\ref{alg:epsapproximateprobeq}) depends only on $\varepsilon, \eta,\delta$ and for the values we have chosen, we find that we require $m=294$ samples. 

Our tests terminated with the correct result in less than 3600 seconds on all the PCs we experimented with.  We present the empirical results in Table~\ref{tab:comparision}. The first column indicates the benchmark's name, the second and third indicate the result and runtime of $\teq$  when presented with a pair of $\varepsilon$-close PCs as input. Similarly, the fourth and fifth columns indicate the result and observed runtime of $\teq$ when the input PCs are $\eta$-far . The full set of results are presented in the supplementary Section~\ref{sec:extended}.

\section{A characterization of the complexity of  testing }\label{sec:characterization}
In this section, we characterize PCs according to the complexity of closeness and equivalence testing. We present the characterization in Table~\ref{tab:tolerant}.
The results presented in the table can be separated into (1) hardness results, and (2) upper bounds. The hardness results, presented in Section~\ref{sec:hardness}, are largely derived from known complexity-theoretic results. The upper bounds, presented in Section~\ref{sec:upperbound}, are derived from a combination of established results, our algorithm $\teq$ and the exact equivalence test of~\citet{DH02}(presented in supplementary Section~\ref{sec:exactequiv} for completeness).

\subsection{Upper bounds}\label{sec:upperbound}
In Table~\ref{tab:tolerant} we label the pair of classes of PCs that admit a poly-time closeness and equivalence test with green symbols $\textcolor{green}{C}$ and $\textcolor{green}{E}$ respectively. \citet{DH02} provided an equivalence test for {\ddnnf} s. From Theorem~\ref{thm:tolapprox}, we know that PCs that supports the $\awct$ and $\asamp$ queries in poly-time must also admit a poly-time approximate equivalence test. A weighted model counting algorithms for {\ddnnf}s was first provided by~\citet{D01a}, and a weighted sampler was provided by~\cite{GSRM19}. ~\citet{ACJR20} provided the first approximate counting and uniform sampling algorithm for {\sdnnf}s. Using the following lemma, we show that with the use of chain formulas, the uniform sampling and counting algorithms extend to log-linear {\sdnnf} distributions as well.
\begin{restatable}{lem}{approxcounting}\label{lem:sdnnf}
	Given a {\sdnnf} formula $\varphi$ (with a v-tree ${T}$), and a weight function $\mathtt{w}$, $\asamp(\varphi,\mathtt{w})$ requires polynomial time in the size of $\varphi$.    
\end{restatable}
The proof is provided in the supplementary Section~\ref{propositions}.

\subsection{Hardness}\label{sec:hardness}
In Table~\ref{tab:tolerant}, we claim that the pairs of classes of PCs labeled with symbols $\textcolor{red}{C}$ and $\textcolor{red}{E}$ , cannot be tested in poly-time for closeness equivalence, respectively. Our claim assumes that the polynomial hierarchy (PH) does not collapse. To prove the hardness of testing the labeled pairs, we combine previously known facts about PCs  and a few new arguments. Summarizing for brevity,
\begin{itemize}
	\item We start off by observing that PC families are in a hierarchy, with {\cnf} $\subseteq $ {\nnf} and {\dnf} $\subseteq$ {\sdnnf} $\subseteq$ {\dnnf} ~\cite{D02}.
	
	\item We then reduce the problem of satisfiability testing of {\cnf}s (NP-hard) and validity testing of {\dnf}s (co-NP-hard) into the problem of equivalence and closeness testing of PCs, in Propositions~\ref{prop:cnfexact}, \ref{prop:cnfapprox} and \ref{prop:dnfexact}. These propositions and their proofs can be found in the supplementary Section~\ref{propositions}.
	
	\item We then connect the existence of poly-time algorithms for equivalence to the collapse of PH via a complexity result due to~\citet{KL80}.
\end{itemize}
\begin{table*}[]
\centering
\begin{tabular}{|c|ccccc|} 
 \hline
 &	\nnf  & \primeimp &\dnnf& \sdnnf&\ddnnf\\
 \hline	
 \nnf&  \textcolor{red}{$EC$}&&&&\\
 \hline
 \primeimp &\textcolor{red}{$EC$}&$UU$&&&\\
 \hline
 \dnnf &\textcolor{red}{$EC$} &\textcolor{red}{$E$}$U$& \textcolor{red}{$E$}$U$&&\\  
 \hline
 \sdnnf&\textcolor{red}{$EC$}&\textcolor{red}{$E$}$U$&\textcolor{red}{$E$}$U$&
 \textcolor{red}{$E$}\textcolor{green}{$C$} &\\  
 \hline
\ddnnf &\textcolor{red}{$EC$}&$UU$& \textcolor{red}{$E$}$U$ & \textcolor{red}{$E$}\textcolor{green}{$C$}& \textcolor{green}{$EC$}\\
\hline 	 
\end{tabular}
\caption{\label{tab:tolerant} Summary of results. \textcolor{green}{C} (resp. \textcolor{green}{E}) indicates that a poly-time closeness (resp. equivalence) test exists. \textcolor{red}{C} (resp. \textcolor{red}{E}) indicates that a poly-time closeness (equivalence) test exists only if PH collapses. `$U$' indicates that the existence of a poly-time test is not known. The table is best viewed in color.}
\end{table*}
\raggedbottom 
The NP-hardness of deciding the equivalence of pairs of {\dnnf}s and pairs of {\sdnnf}s was first shown by~\citet{PD08}. We recast their proofs in the language of distribution testing for the sake of completeness in the supplementary Section~\ref{propositions}.

\section{Conclusion and future work}\label{conclusion}
In this paper, we studied the problem of closeness testing of PCs. Before our work, poly-time algorithms were known only for the special case of equivalence testing of PCs; and, no poly-time closeness test was known for any PC. We provided the first such test, called $\teq$, that used ideas from the field of distribution testing to design a novel algorithm for testing the closeness of PCs. We then implemented a prototype for $\teq$, and tested it on publicly available benchmarks to determine the runtime performance. Experimental results demonstrate the effectiveness of $\teq$ in practice.

We also characterized PCs with respect to the complexity of deciding equivalence and closeness. We combined known hardness results, reductions, and our proposed algorithm $\teq$ to classify pairs of PCs according to closeness and equivalence testing complexity. Since the characterization is incomplete, as seen in Table~\ref{tab:tolerant}, there are questions left open regarding the existence of tests for certain PCs, which we leave for future work.
\section*{Broader Impact}\label{sec:impact}
Recent advances in probabilistic modeling techniques have led to increased adoption of the said techniques in safety-critical domains, thus creating a need for appropriate verification and testing methodologies. This paper seeks to take a step in this direction and focuses on testing properties of probabilistic models likely to find use in safety-critical domains. Since our guarantees are probabilistic,  practical adoption of such techniques still requires careful design to handle failures.
\begin{ack}
	We are grateful to the anonymous reviewers of UAI 2021 and NeurIPS 2021 for their constructive feedback that greatly improved the paper. We would also like to thank Suwei Yang and Lawqueen Kanesh for their useful comments on the earlier drafts of the paper. This work was supported in part by National Research Foundation Singapore under its NRF Fellowship Programme[NRF-NRFFAI1-2019-0004 ] and AI Singapore Programme [AISG-RP-2018-005],  and NUS ODPRT Grant [R-252-000-685-13]. The computational work for this article was performed on resources of the National Supercomputing Centre, Singapore (\href{https://www.nscc.sg}{https://www.nscc.sg}).
\end{ack}
\newpage
\bibliographystyle{plainnat}
\bibliography{neurips_2021}
\appendix
\clearpage
\section{Proofs omitted from the paper}\label{appendix}
\subsection{A test for equivalence}\label{sec:exactequiv}
For the sake of completeness we recast the {\ddnnf} circuit equivalence test of~\citet{DH02} into an equivalence test for log-linear probability distributions.
\begin{algorithm}[]
	\caption{$\peq(\varphi_1,\mathtt{w_1}, \varphi_2,\mathtt{w_2}, \delta)$}
	\label{alg:exactprobeq}
	\begin{algorithmic}[1]
		\STATE $m \gets \lceil n/  \delta\rceil $ 
		\STATE $  \theta \sim [m]^n $
		\IF{$\pi(\varphi_1,\mathtt{w_1})(\theta) = \pi(\varphi_2, \mathtt{w_2})(\theta)$}\alglinelabel{line:compare}
		\STATE \textbf{Return} \accept
		\ELSE \STATE \textbf{Return} \reject 
		\ENDIF
	\end{algorithmic}
\end{algorithm}
\raggedbottom
\paragraph{The algorithm:} The pseudocode for $\peq$ is shown in Algorithm~\ref{alg:exactprobeq}. $\peq$ takes as input two satisfiable circuits $\varphi_1, \varphi_2$ defined over $n$ Boolean variables, a pair of weight functions $\mathtt{w_1}, \mathtt{w_2}$ and a tolerance parameter $\delta \in (0,1)$. Recall that a circuit $\varphi$ and a weight function $\mathtt{w}$ together define the probability distribution $P(\varphi, \mathtt{w})$. $\peq$ returns {\accept} with confidence 1 if the two probability distributions $P(\varphi_1, \mathtt{w_1})$ and $P(\varphi_2, \mathtt{w_2})$ are equivalent, i.e. $\dtv(P(\varphi_1, \mathtt{w_1}), P(\varphi_2, \mathtt{w_2})) = 0$. If  $\dtv(P(\varphi_1, \mathtt{w_1}), P(\varphi_2, \mathtt{w_2})) > 0$, then it returns {\reject} with confidence at least $1-\delta$. 

The algorithm starts by drawing a uniform random assignment $\theta$ from $[m]^n$, where $m=\lceil n/\delta\rceil$.  Using the procedure given in Proposition~\ref{prop:time} (in Section~\ref{appendix:exact}), $\peq$ computes the values $\pi(\varphi_1, \mathtt{w_1})(\theta)$ and $\pi(\varphi_2, \mathtt{w_2})(\theta)$, where $\pi(\varphi, \mathtt{w})$ is the network polynomial~\cite{D03}. $\pi(\varphi, \mathtt{w})$ defined as:
\begin{align*}
	\pi(\varphi, \mathtt{w}) = \sum_{\sigma \in R_\varphi}  \frac{\mathtt{w}(\sigma)}{\mathtt{w}(\varphi)} \left( \prod_{x_i \models \sigma} x_i \prod_{\lnot x_j \models \sigma}(1 - x_j)\right)
\end{align*}

The two values are then compared on line~\ref{line:compare}, and if they are equal the algorithm returns {\accept} and otherwise  returns {\reject}. The central idea of the test is that whenever the two distributions $P(\varphi_1, \mathtt{w_1})$ and $P(\varphi_2, \mathtt{w_2})$ are equivalent, the polynomials $\pi(\varphi_1, \mathtt{w_1})$ and $\pi(\varphi_2, \mathtt{w_2})$ are also equivalent, however when they are not equivalent, the polynomials disagree on atleast $1-\delta$ fraction of assignments from the set $[m]^n$.

We formally claim and prove the correctness of $\peq$ in Lemma~\ref{lem:exact} in the Section~\ref{appendix:exact}.
\subsection{An analysis for Algorithm~\ref{alg:exactprobeq}}\label{appendix:exact}
In this section, we present the theoretical analysis of Algorithm \ref{alg:exactprobeq} ($\peq$) and the proof of the following lemma.
\begin{restatable}{lem}{exact}\label{lem:exact}
	Given two satisfiable probabilistic circuits $\varphi_1, \varphi_2$ and weight functions $\mathtt{w_1},\mathtt{w_2}$, along with confidence parameter $\delta \in (0,1)$.
	\begin{enumerate}
		\item[A.] If $\dtv(P(\varphi_1,\mathtt{w_1}), P(\varphi_2,\mathtt{w_2})) = 0$, then $\peq(\varphi_1,\mathtt{w_1}, \varphi_2,\mathtt{w_2},  \delta)$ returns {\accept} with probability 1.
		\item[B.] If $\dtv(P(\varphi_1,\mathtt{w_1}), P(\varphi_2,\mathtt{w_2})) > 0$, then $\peq(\varphi_1,\mathtt{w_1}, \varphi_2,\mathtt{w_2},\delta)$ returns {\reject} with probability at least $ (1-\delta)$.
	\end{enumerate} 
\end{restatable}
$\peq$ returns {\accept} if  $\pi(\varphi_1, \mathtt{w_1})(\sigma) = \pi(\varphi_2, \mathtt{w_2})(\sigma)$. Since $P(\varphi_1, \mathtt{w_1}) \equiv P(\varphi_2, \mathtt{w_2}) \rightarrow \pi(\varphi_1, \mathtt{w_1}) \equiv \pi(\varphi_2, \mathtt{w_2})$, it follows that $\peq$  always returns {\accept} for two equivalent probabilistic distributions.

For the proof of Lemma~\ref{lem:exact}(B) we will first define some notation, and then we show (in Lemma~\ref{lem:noteq}) that a random assignment over $[m]^n$ is likely to be a witness for non-equivalence with probability $>1-\delta$. The proof immediately follows as we know that $\peq$ returns {\reject} if $\pi(\varphi_1,\mathtt{w_1})(\sigma) \not = \pi(\varphi_2, \mathtt{w_2})(\sigma)$. 

\begin{definition}
	$\pi|_{x_i = 1}(\varphi,\mathtt{w})$ is a polynomial over $n-1$ variables, obtained by setting the variable $x_i$ to 1. Similarly $\pi|_{x_i = 0}(\varphi,\mathtt{w})$ is  obtained by setting the variable $x_i$ to 0, thus:  
	\begin{align*}
		\pi(\varphi,\mathtt{w}) = (1 - x_i) \pi|_{x_i = 0}(\varphi,\mathtt{w}) + x_i \pi|_{x_i = 1}(\varphi,\mathtt{w})
	\end{align*}
\end{definition}

From the definition, we can immediately infer the following proposition.
\begin{proposition}\label{prop:split}
	If $\pi(\varphi_1, \mathtt{w_1}) \not \equiv \pi(\varphi_2, \mathtt{w_2})$ then  for all $x_i$, at least one of the following must be true:
	\begin{itemize}
		\item $\pi|_{x_i = 1}(\varphi_1, \mathtt{w_1}) \not = \pi|_{x_i = 1}(\varphi_2, \mathtt{w_2})$ 
		\item $\pi|_{x_i = 0}(\varphi_1, \mathtt{w_1}) \not = \pi|_{x_i = 0}(\varphi_2, \mathtt{w_2})$
	\end{itemize}
\end{proposition}

For the proofs in this section, we will use the following notation. For a  circuit $\varphi$  defined over the variables $\{x_1,\ldots, x_n\}$, we define a polynomial $P(\varphi, \mathtt{w}): \{0,1\}^n \rightarrow [0,1]$:
\begin{align*}
	P(\varphi, \mathtt{w}) = \sum_{\sigma \in R_\varphi}  \frac{\mathtt{w}(\sigma)}{\mathtt{w}(\varphi)} \left( \prod_{x_i \models \sigma} x_i \prod_{\lnot x_j \models \sigma}(1 - x_j)\right)
\end{align*}

We define another polynomial $\pi(\varphi, \mathtt{w})$ which is  $P(\varphi, \mathtt{w})$ but defined from $[m]^n \rightarrow \mathbb{Q}$ where $[m] = \{1\ldots,m\}$.  

To show that the polynomial $\pi(\varphi, \mathtt{w})$ can be computed in time polynomial in the size of the representation, we will adapt the procedure given by~\cite{DH02}.
\begin{proposition}\label{prop:time}
	Let $\varphi$ be a circuit over the set $X=\{x_1, \ldots, x_n\}$ of $n$ variables , that admits poly-time WMC. Let $\mathtt{w}:X\rightarrow \mathbb{Q}^+$ be a weight function and let $\theta \in [m]^n$ be an assignment to the variables in $X$ and $\theta(x)$ be the assignment to variable $x \in X$ in $\theta$. For each node $\eta$ in the circuit, define a function ${S}()$ recursively as follows:
	\begin{itemize}
		\item ${S}(\eta) = \sum_i {S}(n_i)$, where $\eta$ is an or-node with children $n_i$.
		\item ${S}(\eta) = \prod_i {S}(n_i)$, where $\eta$ is an and-node with children $n_i$.
		\item ${S}(\eta) = \begin{cases}
			0,&\text{if }\eta \text{ is a leaf node false}\\
			1,&\text{if } \eta \text{ is a leaf node true}\\
			\mathtt{w}(x)\theta(x),&\text{if } \eta \text{ is a leaf node } x \in X\\
			(1-\mathtt{w}(x))(1-\theta(x)),&\text{if } \eta \text{ is a leaf node } \lnot x, x \in X\\
		\end{cases}$
		\item $\pi(\varphi,\mathtt{w}) = {S}(\eta)/\mathtt{w}(\varphi) $, \text{ where $\eta$ is the root node}
	\end{itemize}
	We can compute the quantity $\mathtt{w}(\varphi) $ in linear time due to our assumption of poly-time WMC, hence we can find $\pi(\varphi,\mathtt{w})(\theta)$ in time linear in the size of the {\ddnnf}.
\end{proposition}

\begin{restatable}{lem}{noteq}\label{lem:noteq}
	For a random assignment $\sigma \sim [m]^n$, $\Pr[\pi(\varphi_1,{\mathtt{w_1}})(\sigma) \not = \pi(\varphi_2,\mathtt{w_2})(\sigma) \mid P(\varphi_1,\mathtt{w_1}) \not \equiv P(\varphi_2,\mathtt{w_2})  ] >1 - \delta $
\end{restatable}
\begin{proof}
	For $n=1$, $\sigma$ is an assignment to a single variable $x$. The polynomial on the single variable $x$ can be parameterised as $\pi(\varphi, \mathtt{w})(x) = \alpha x + (1-\alpha) (1-x)$ where parameter $\alpha = \frac{\mathtt{w(x)}}{\sum_{\theta \in R_{\varphi}}\mathtt{w(\theta)}}$. Let polynomials $\pi(\varphi_1, \mathtt{w_1}), \pi(\varphi_2, \mathtt{w_2})$ be parameterised with $\alpha_1, \alpha_2$, respectively. Our assumption that $ P(\varphi_1,\mathtt{w_1}) \not \equiv P(\varphi_2,\mathtt{w_2})$ immediately leads to the fact that $\pi(\varphi_1, \mathtt{w_1}) \not \equiv \pi(\varphi_2, \mathtt{w_2})$ which in turn implies that $\alpha_1 \not = \alpha_2$.
	
	The the set of inputs $x$ for which two non-equivalent polynomials agree is given by,
	\begin{align*}
		\pi(\varphi_1, \mathtt{w_1})(x) &= \pi(\varphi_2, \mathtt{w_2})(x) \\
		\alpha_1 x + (1-\alpha_1) (1-x) &= \alpha_2 x + (1-\alpha_2) (1-x)\\
		2(\alpha_1-\alpha_2)x  &= \alpha_1 - \alpha_2\\
		x&= 1/2 
	\end{align*}
	
	From the initial assumption we know that $x$ can only take integer values, hence there are no inputs in the set $[m]$ for which $\pi(\varphi_1, \mathtt{w_1})(\sigma) \not = \pi(\varphi_2, \mathtt{w_2})(\sigma)$. Thus, for $n=1$, and any $\sigma$, $\Pr[\pi(\varphi_1,{\mathtt{w_1}})(\sigma) \not = \pi(\varphi_2,\mathtt{w_2})(\sigma) \mid P(\varphi_1,\mathtt{w_1}) \not \equiv P(\varphi_2,\mathtt{w_2})  ]=0$
	
	We now assume that the hypothesis holds for $n-1$ variables. Consider polynomials $\pi(\varphi_1, \mathtt{w_1}) \not \equiv \pi(\varphi_2, \mathtt{w_2})$ over $n$ variables. From  Prop~\ref{prop:split} we know that at least one of the following holds:
	\begin{itemize}
		\item $\pi|_{x_i = 1}(\varphi_1, \mathtt{w_1}) \not = \pi|_{x_i = 1}(\varphi_2, \mathtt{w_2})$
		\item $\pi|_{x_i = 0}(\varphi_1, \mathtt{w_1}) \not = \pi|_{x_i = 0}(\varphi_2, \mathtt{w_2})$
	\end{itemize}
	Without any loss of generality we assume the latter. Then we know that there exists a set $\Sigma \subseteq [m]^{n-1}, |\Sigma| \geq (m-1)^{n-1}$, such that 
	\begin{align*}
		\forall_{\sigma \in \Sigma}, \pi|_{x_n = 0}(\varphi_1, \mathtt{w_1})(\sigma) \not = \pi|_{x_n = 0}(\varphi_2, \mathtt{w_2})(\sigma)
	\end{align*}
	
	The set of assignments $\sigma$ for which $\pi(\varphi_1, \mathtt{w_1})(\sigma) = \pi(\varphi_2, \mathtt{w_2})(\sigma)$ is given by,
	\begin{align*}
		\pi(\varphi_1, \mathtt{w_1})(\sigma) &= \pi(\varphi_2, \mathtt{w_2})(\sigma)
		\\
		(1 - x_{n}) \pi|_{x_{n} = 0}(\varphi_1,\mathtt{w_1})(\sigma) + x_{n}\pi|_{x_{n} = 1}(\varphi_1,\mathtt{w_1})(\sigma) &= (1 - x_{n}) \pi|_{x_{n} = 0}(\varphi_2,\mathtt{w_2})(\sigma) + x_{n} \pi|_{x_{n} = 1}(\varphi_2,\mathtt{w_2})(\sigma)
	\end{align*}
	\begin{align*}
		x_{n} (\pi|_{x_{n} = 1}(\varphi_1,\mathtt{w_1})(\sigma) - \pi|_{x_{n} = 0}(\varphi_1,\mathtt{w_1})(\sigma) - \pi|_{x_{n} = 1}(\varphi_2,\mathtt{w_2})(\sigma)
		+ \pi|_{x_{n} = 0}(\varphi_2,\mathtt{w_2})(\sigma))\\
		= \pi|_{x_{n} = 0}(\varphi_2,\mathtt{w_2})(\sigma) - \pi|_{x_{n} = 0}(\varphi_1,\mathtt{w_1})(\sigma)
	\end{align*}
	From the assumptions we know that there are at least $(m-1)^{n-1}$ assignments $\sigma$ s.t.  $ \pi|_{x_{n} = 0}(\varphi_2,\mathtt{w_2})(\sigma) - \pi|_{x_{n} = 0}(\varphi_1,\mathtt{w_1})(\sigma) \not = 0$, from which we can conclude that the RHS is non-zero. Thus for all such $\sigma$ there can be at most one value of $x_n$ for which the equality holds, which leaves $m-1$ values which $x_n$ cannot take. Thus there are at least $(m-1)\times (m-1)^{n-1} = (m-1)^{n}$ assignments to $n$ variables for which $\pi(\varphi_1, \mathtt{w_1})(\sigma) \not =  \pi(\varphi_2, \mathtt{w_2})(\sigma)$. 
	
	Since the total number of assignments for $n$ variables is $m^n$, out of which $(m-1)^n$ witness the non-equivalence of the two probability distributions, we know that for a randomly chosen assignment $\sigma \sim [m]^n $, we have
	\begin{align*}
		\Pr[\pi(\varphi_1,{\mathtt{w_1}})(\sigma) \not = \pi(\varphi_2,\mathtt{w_2})(\sigma) \mid  P(\varphi_1,\mathtt{w_1}) \not \equiv P(\varphi_2,\mathtt{w_2})  ] 	&\geq \frac{(m-1)^n}{m^n} \geq \left( 1- \frac{ \delta }{n} \right)^n  \\&> 1- \delta \text{\quad(using $m$ from Algorithm~\ref{alg:exactprobeq})}	
	\end{align*}
\end{proof}

\clearpage
\subsection{Omitted proof from the analysis of Algorithm~\ref{alg:epsapproximateprobeq}}\label{sec:approx}

In this subsection, we present the proof of  Theorem~\ref{thm:tolapprox}(B), and Theorem~\ref{thm:runtime}. Recall that we use $P_1$ and $P_2$ to refer to $P(\varphi_1, \mathtt{w_1})$ and $P(\varphi_2, \mathtt{w_2})$, respectively.

\subsection{Proof of Lemma~\ref{lem:allerrors}}\label{sec:proof9}
\begin{proof}
	The quantity $r(\sigma)$ (line~\ref{line:r2} from Algorithm~\ref{alg:epsapproximateprobeq}) conditioned on the event $\overline{\mathtt{Fail_i}} \subset \mathtt{Good}$:
	\begin{align*}
		r(\sigma) = \frac{\mathtt{w_2}(\sigma)}{\awct(\sqrt{1+\gamma/4}-1,\delta/8,\varphi_2,\mathtt{w_2})} \cdot 
		\frac{\awct(\sqrt{1+\gamma/4}-1,\delta/8,\varphi_1, \mathtt{w_1})}{\mathtt{w_1}(\sigma)}
	\end{align*}
	Conditioned on the events $\mathtt{Pass_1}, \mathtt{Pass_2} \subset \mathtt{Good}$, we know that with probability 1:
	\begin{align*}
		\frac{\mathtt{w_2}(\sigma)\mathtt{w_1}(\varphi_1)}{\mathtt{w_2}(\varphi_2)\mathtt{w_1}(\sigma)} (\sqrt{1+\gamma/4})^{-2} < 	r(\sigma) 
		<(\sqrt{1+\gamma/4})^2  \frac{\mathtt{w_2}(\sigma)\mathtt{w_1}(\varphi_1)}{\mathtt{w_2}(\varphi_2)\mathtt{w_1}(\sigma)}  
	\end{align*}
	Which gives us: $\frac{P_2(\sigma)}{P_1(\sigma)}  (1+\gamma/4)^{-1}<r(\sigma) < (1+\gamma/4) \frac{P_2(\sigma)}{P_1(\sigma)}$ and therefore,\\
	$		\left| r(\sigma)-\frac{P_2(\sigma)}{P_1(\sigma)} \right| \leq   \frac{P_2(\sigma)}{P_1(\sigma)} \cdot   \underset{0<\gamma<1}{max}\left(\gamma/4, 1-\frac{1}{1+\gamma/4}\right)\leq  \frac{P_2(\sigma)}{P_1(\sigma)} \cdot \gamma/4 $
\end{proof}
\subsubsection{Proof of Theorem~\ref{thm:tolapprox}(A)}\label{sec:proof8}
\begin{proof}
	We assume the event $\mathtt{Good}$. Let $\sigma_i$ be the sample returned by the sampler $\asamp$ in the $i$th iteration.  
	If $r(\sigma_i)>1$, $\Gamma[i] $ takes value 0, else $\Gamma[i]=1-r(\sigma_i)$. 	Thus $\Gamma[i]$ is a r.v. which takes on a value from $[0,1]$. We can write $\Gamma[i] = \mathbbm{1}\left(r(\sigma_i)<1\right)\left(1- r(\sigma_i)\right)$	
	The expectation of $\Gamma[i]$ is:
	\begin{align}
		&\mathbbm{E}[\Gamma[i]] =\sum_{\sigma \in \{0,1\}^n} \mathbbm{1}\left(r(\sigma)<1\right)\left(1- r(\sigma)\right)\cdot \Pr[\asamp(\gamma/(4\eta- 2\gamma),\delta/4m,\varphi_1,\mathtt{w_1}) = \sigma] \label{line:expectationexpansion}
	\end{align}
	According to definition~\ref{def:asamp}, and our assumption of $\overline{\mathtt{Fail_i}} \subset \mathtt{Good}$, we know that with probability 1
	$\Pr[\asamp(\gamma/(4\eta - 2\gamma), \delta/4m, \varphi_1 , \mathtt{w_1}) = \sigma]\leq (1+\gamma/(4\eta - 2\gamma))P_1(\sigma)
	$. Thus we have,
	\begin{align*}
		\mathbbm{E}[\Gamma[i]]
		&\leq  \sum_{\sigma \in \{0,1\}^n} \mathbbm{1}\left(r(\sigma)<1\right)\left(1- r(\sigma)\right)\cdot (1+\gamma/(4\eta-2\gamma))P_1(\sigma)
	\end{align*}
	Recall that in Lemma~\ref{lem:boundonA}, we define $A = \sum_{\sigma \in \{0,1\}^n} \mathbbm{1}\left(r(\sigma)<1\right)\left(1 -r(\sigma)\right) P_1(\sigma)$. Therefore, we can simplify the above expression as: $\mathbbm{E}[\Gamma[i]] = (1+\gamma/(4\eta-2\gamma)) \cdot  A$. We can then use the assumption of $\varepsilon$-closeness and the result of Lemma~\ref{lem:boundonA}-1 to find a bound on the expectation,
	\begin{align*}
		\mathbbm{E}[\Gamma[i]]& \leq\left(1+\gamma/(4\eta-2\gamma)\right)(\varepsilon + \gamma/4)
		\leq \varepsilon + \gamma/2
	\end{align*}
	Using  the linearity of expectation we get:
	$		\mathbbm{E}\left[\sum_{i \in [m]} \Gamma[i]\right] <   m(\varepsilon + \gamma/2 ) $. $\teq$ returns {\reject} when $\sum_{i \in [m]}\Gamma[i]>m(\varepsilon+\gamma)$ on line~\ref{line:if2}. Since the $\Gamma[i]$’s are i.i.d random variables taking values in [0, 1], we apply the Chernoff bound to find the probability of {\accept}, assuming the event $\mathtt{Good}$:
	\begin{align*}
		\Pr\left[\teq\text{ returns {\accept}} ~\bigg|~  \mathtt{Good} \right] &=1 - \Pr\left[ \sum_{i \in [m]} \Gamma[i]>m(\varepsilon+\gamma) \right]  
		&\geq  1 - 2e^{-\gamma^2m/2}  \geq  1- \delta/2   
	\end{align*}
	The value for $m$ is taken from line \ref{line:m2} of Algorithm~\ref{alg:epsapproximateprobeq}. Using ~(\ref{line:excludebadevents}), we see that the probability of $\teq$ returning {\accept} is: 	$\Pr[\teq\text{ returns {\accept}}] \geq   \Pr[\teq\text{ returns {\accept}} \mid   \mathtt{Good}] \Pr[ \mathtt{Good}] =(1-\delta/2)(1-\delta/2) \geq 1 - \delta$ 
\end{proof}
\subsubsection{Proof of Theorem~\ref{thm:tolapprox}(B)}\label{sec:proof6}
\begin{proof}
	First we assume the event $\mathtt{Good}$. Then according to definition ~\ref{def:asamp}, we know that with probability 1 (since we assume event $\overline{\mathtt{Fail_i}}\subset \mathtt{Good}$)
	\begin{align*}
		\Pr[\asamp(\gamma/(4\eta - 2\gamma), \delta/4m, \varphi_1 , \mathtt{w_1}) = \sigma]     \geq \frac{P_1(\sigma)}{(1+\gamma/(4\eta -2\gamma))}  
	\end{align*} 
	Thus substituting into (\ref{line:expectationexpansion}), we get 
	\begin{align}
		\mathbbm{E}[\Gamma[i]]  \geq \sum_{\sigma \in \{0,1\}^n}\mathbbm{1}\left(r(\sigma)<1\right)(1-r(\sigma))\frac{P_1(\sigma_i)}{1+\gamma/(4\eta -2\gamma)} \label{line:boundonexpectationeta}
	\end{align}
	Then we use the $\eta$-farness assumption and Lemma~\ref{lem:boundonA}-2 
	\begin{align}	
		\mathbbm{E}[\Gamma[i]] &\geq \frac{\eta - \gamma/4}{1+\gamma/(4\eta -2\gamma)} = \eta - \gamma/2 \label{line:worstcaseexpectationeta}
	\end{align}
	The algorithm returns {\accept} when $\sum_{i \in [m]}\Gamma[i]\leq m(\varepsilon+\gamma)$ (on line~\ref{line:if2}). Then using ~(\ref{line:worstcaseexpectationeta}) and the linearity of expectation.
	\begin{align*}
		\mathbbm{E}\left[\sum_{i \in [m]} \Gamma[i]]\right] &\geq  m(\eta- \gamma/2) 
	\end{align*}
	Since the $\Gamma[i]$’s are i.i.d random variables taking values in [0, 1], we apply the Chernoff bound to find the probability of {\reject}, given the assumption of the event $\mathtt{Good}$:
	
	\begin{align}
		\Pr\left[\teq\text{ returns {\reject}} \mid \mathtt{Good}\right]&= 1 - \Pr\left[ \sum_{i \in [m]} \Gamma[i] \leq m(\varepsilon+\gamma) \right] \nonumber 
		\\
		&\geq  1 - \Pr\left[m(\eta-\gamma/2) - \sum_{i \in [m]} \Gamma[i]   \geq    m(\eta-\gamma/2-\varepsilon-\gamma) \right] \nonumber \\
		&\geq  1 - \Pr\left[|\sum_{i \in [m]} \Gamma[i] - m(\eta-\gamma/2)  | \geq m\gamma/2 \right] \nonumber \\
		&\geq  1 - 2e^{-\gamma^2m/2}  \geq  1- \delta/2 \text{\quad (Substituting $m$ as in line \ref{line:m2})} \nonumber 
	\end{align}
	
	Hence, the probability that Algorithm~\ref{alg:epsapproximateprobeq} returns {\reject} is 
	\begin{align*}  
		\Pr[\teq\text{ returns {\reject}}]	
		&\geq \Pr\left[\teq\text{ returns {\reject}} \mid \mathtt{Good}\right]\Pr\left[\mathtt{Good}\right]  \\ 
		& =(1-\delta/2)(1-\delta/2) \geq 1 - \delta \quad \text{(Using~(\ref{line:excludebadevents}))}
	\end{align*} 
\end{proof}

\subsubsection{Proof of Theorem~\ref{thm:runtime}}\label{sec:proof7}
\begin{proof}
	$\teq$ makes two calls to $\awct$ on line~\ref{line:awct1} and \ref{line:awct2} of Algorithm~\ref{alg:epsapproximateprobeq}. According to definition~\ref{def:awct}, the runtime of the  $\awct(\sqrt{1+\gamma/4} -1, \delta/8,\varphi,\mathtt{w})$ query is $T(\sqrt{1+\gamma/4}-1,\delta/8,\varphi) = poly((\sqrt{1+\gamma/4} -1)^{-1}, \log(\delta^{-1}),|\varphi|)$. 
	
	Using the identity $1+\frac{x}{2} - \frac{x^2}{2} \leq \sqrt{1+x}$ for $x\geq 0$ and the fact that $\gamma \in (0,1)$
	\begin{align*}
		\frac{1}{\sqrt{1+\gamma/4}-1} \leq 	\frac{1}{\gamma/8 - \gamma^2/32}  <\frac{11}{\gamma} 
	\end{align*}
	Hence any $poly((\sqrt{1+\gamma/4} - 1)^{-1})$ algorithm   also runs in $poly(\gamma^{-1})$. Thus the $\awct$ queries run in $O(poly(\gamma^{-1}, \log(\delta^{-1}), max(|\varphi_1|,|\varphi_2|)))$

	$\teq$ makes $m =\lceil \log(2/\delta)/2\gamma^2\rceil $ calls to $\asamp$ on lines~\ref{line:asamp} of Algorithm~\ref{alg:epsapproximateprobeq}. According to definition~\ref{def:asamp}, the runtime of the  $\asamp(\gamma/(4\eta-2\gamma), \delta/4m,\varphi_1,\mathtt{w_1})$ query is $T(\gamma/(4\eta-2\gamma),\delta/4m,|\varphi_1|) = poly((\gamma/(4\eta-2\gamma))^{-1}, \log((\delta/4m)^{-1}),|\varphi_1|)$. First we see that $\frac{4\eta - 2 \gamma}{\gamma} < \frac{4}{\gamma}$, thus the algorithm remains in $poly(\gamma^{-1})$. We then see that $\log(4m/\delta) = \log(4m) + \log(\delta^{-1})$. Since $\log(m) \in poly(\log(\gamma^{-1}),\log\log(\delta^{-1}))$, we know that $\asamp$ queries run in $O(poly(\gamma^{-1}, \log(\delta^{-1}), max(|\varphi_1|,|\varphi_2|)))$.
	
	Since each $\asamp$ call and each $\awct$ call requires atmost polynomial time in terms of $\gamma^{-1},\log(\delta^{-1})$ and $max(|\varphi_1|,|\varphi_2|)$ we know that the algorithm itself runs in time polynomial in  $\gamma^{-1},\log(\delta^{-1})$ and $max(|\varphi_1|,|\varphi_2|)$.
\end{proof}

\subsection{Proofs omitted from Section~\ref{sec:characterization}}\label{propositions}
For the following proofs, we assume a uniform weight function.  

\begin{restatable}{prop}{proposition}\label{prop:cnfexact}
	If there exists a poly-time randomised algorithm for deciding the equivalence of a pair of PCs with at least one PC in {\cnf}, then NP=RP.
\end{restatable}
\begin{proof}
	For CNFs, testing satisfiability is known to be NP-hard. Consider a CNF $\varphi$ defined over variables $\{x_1,\ldots,x_n\}$ and a circuit $\psi$ s.t. $\psi \equiv \bigwedge_{i\in [n+1]}x_i$. Define
	\begin{align*}
		\hat\varphi = (\lnot x_{n+1} \rightarrow \varphi) \wedge ( x_{n+1} \rightarrow \bigwedge_{i\in[n]}x_i) 
	\end{align*}
	We see that the size of the new CNF is  $|\hat\varphi|\in O(|\varphi| + n)$.  $\hat\varphi$ has at least one satisfying assignment, specifically the assignment $\forall_{i\in [n+1]}x_i = 1$. We notice that $\dtv(P(\hat\varphi, \mathtt{w}), P(\psi, \mathtt{w}))=0$ if and only if $|R_{\varphi}|=0$.  Thus the existence of a poly-time randomised algorithm for deciding whether $\dtv(P(\hat\varphi, \mathtt{w}), P(\psi, \mathtt{w}))=0$ would  imply NP $\subseteq$ RP and hence NP=RP.
\end{proof}

\begin{restatable}{prop}{proposition}\label{prop:cnfapprox}
	If there exists a poly-time randomised algorithm for deciding the closeness of a pair of PCs with at least one PC in {\cnf}, then NP=RP.
\end{restatable}
\begin{proof}
	$\dtv(P(\hat\varphi, \mathtt{w}), P(\psi, \mathtt{w}))\geq 0.5$ if and only if $|R_{\varphi}|>0$.	Assume there exists a poly-time randomised algorithm which returns {\reject} if $\dtv(P(\hat\varphi, \mathtt{w}), P(\psi, \mathtt{w}))\geq 0.4$ and {\accept} if $\dtv(P(\hat\varphi, \mathtt{w}), P(\psi, \mathtt{w}))\leq  0.1$ with probability $> 2/3 $. Such an algorithm would imply BPP $\subseteq$ NP, and hence NP=RP.
\end{proof}

\begin{restatable}{prop}{proposition}\label{prop:dnfexact}
	If there exists a poly-time randomised algorithm for deciding the equivalence of a pair of PCs with at least one PC in {\dnf}, then NP=RP.
\end{restatable}
\begin{proof}
	For DNFs, deciding validity is known to be co-NP-hard. Given DNF $\varphi$ and a circuit $\psi = True$, the existence of a poly-time randomized algorithm for checking the equivalence of $\psi$ and $\varphi$  would imply that co-NP $\subseteq$ co-RP and hence co-NP = co-RP. 
\end{proof}
Using Corollary 6.3 from \cite{KL80} we see that PH collapses as a result of either of the above implications.

From the set inclusions \dnf $\subseteq $ \sdnnf $\subseteq $ {\dnnf} and  
\cnf $\subseteq $ \nnf, we obtain all hardness results. From the fact that {\ddnnf}s support weighted counting and sampling we have the existence results.

The following lemma supports our claim in table \ref{tab:tolerant}. 
\approxcounting
\begin{proof}
	Here we will assume that the weights are in the dyadic form i.e. they can be represented as the fraction $d/2^p$ for $d,p \in \mathbb{Z}^+$. Then using the weighted to unweighted construction from~\cite{CFMV15}, the problem of approximate weighted sampling over {\sdnnf} can be reduced to approximate uniform sampling. Given a {\sdnnf} $\varphi$, and a weight function $\mathtt{w}$, we generate a {\sdnnf} $\varphi_w \equiv \varphi \wedge \bigwedge_{i\in[n]}(\lnot x_i \vee C^1_i) \wedge \bigwedge_{i \in [n]}( x_i \vee C^0_i)$. Here, $C^0_i$ is chain formula having exactly $w(\lnot x_i)\times 2^p = 2^p-d$ satisfying assignments, and $C^1_i$ is a chain formula with   $\mathtt{w}(x_i) \times 2^p = d$ satisfying assignments.

	The property of decomposability on the $\wedge$ nodes of $\varphi$ is preserved as each $C_i$ introduces a new set of variables disjoint from the set of variables in $\varphi$ and and also from all $C_j$, such that $j\not = i$. The $\wedge$ nodes in the chain formula are also trivially decomposable and structured as each chain formula variable appears exactly once in the formula.
	
	If $\sigma$ is an assignment to the set of variables of $S$ and if $S' \subseteq S$, then let $\sigma_{\downarrow S'}$ denote the projection of $\sigma$ on the variables in $S'$. The weighted formula $\varphi$ is defined over variable set $var(\varphi)$. The formula $\varphi_w$ defined above has the property that if $ \varphi(\sigma) = 1$, then $ |\{\sigma'|\varphi_w(\sigma')=1 \wedge \sigma'_{\downarrow var(\varphi)} = \sigma \}|/|R_{\varphi_w}|=\mathtt{w}(\sigma) $. Thus a uniform distribution on $R_{\varphi_w}$, when projected on $var(\varphi)$ induces the weighted distribution $P(\varphi,\mathtt{w})$. This property allows weighted sampling and counting on $\varphi$  with the help of a uniform sampler for the generated formula $\varphi_w$.\end{proof}

\clearpage
\section{Experimental evaluation}\label{sec:extended}
In this section we will first discuss the method for generating the synthetic dataset, and then we present the extended table of results.
\subsection{One variable perturbation}\label{sec:syntheticbenchmark}
Consider two weight functions $\mathtt{w_1}$ and $\mathtt{w_2}$ that differ only in the weight assigned to the literals $v^0$ and $v^1$. Then, from the definition of $\dtv$:

\begin{align*}
	\dtv(P(\varphi, \mathtt{w_1}),P(\varphi, \mathtt{w_2})) =\frac{1}{2} \sum_{\sigma \in \{0,1\}^n} \left|\frac{\mathtt{w_1}(\sigma)}{\mathtt{w_1}(\varphi)} - \frac{\mathtt{w_2}(\sigma)}{\mathtt{w_2}(\varphi)}\right|
\end{align*}

Let $S \subseteq \{0,1\}^n$ be the set of assignments for which $\frac{\mathtt{w_1}(\sigma)}{\mathtt{w_1}(\varphi)} > \frac{\mathtt{w_2}(\sigma)}{\mathtt{w_2}(\varphi)}$. Thus,
\begin{align*}
	\dtv (P(\varphi, \mathtt{w_1}),P(\varphi, \mathtt{w_2})) =\sum_{\sigma \in S}\left( \frac{\mathtt{w_1}(\sigma)}{\mathtt{w_1}(\varphi)} - 
	\frac{\mathtt{w_2}(\sigma)}{\mathtt{w_2}(\varphi)} \right)
\end{align*}
Lets assume wlog that $\mathtt{w_1}$ assigns a larger weight to $v^1$ than $\mathtt{w_2}$  does. Then, $S$ contains all and only those assignments that have literal $v^1$, i.e. $S \equiv \varphi \wedge v^1$. Thus, 
\begin{align*}
	\dtv (P(\varphi, \mathtt{w_1}),P(\varphi, \mathtt{w_2}))= \frac{\mathtt{w_1}(\varphi \wedge v^1)}{\mathtt{w_1}(\varphi)} - \frac{\mathtt{w_2}(\varphi \wedge v^1)}{\mathtt{w_2}(\varphi)} 
\end{align*}
We can rewrite $\mathtt{w_1}(\varphi \wedge v^1) = \mathtt{w_1'}(\varphi) \times \mathtt{w_1}(v^1)$, where $\mathtt{w_1'}$ is $\mathtt{w_1}$ with the weight of $v^1$ set to 1. Using a similar transformation on $\mathtt{w_2}(\varphi \wedge v^1)$ we get 
\begin{align*}
	\dtv (P(\varphi, \mathtt{w_1}),P(\varphi, \mathtt{w_2}))= \frac{\mathtt{w_1'}(\varphi) \times \mathtt{w_1}(v^1)}{\mathtt{w_1}(\varphi)} - \frac{\mathtt{w_2'}(\varphi ) \times \mathtt{w_2}(v^1)}{\mathtt{w_2}(\varphi)}
\end{align*}
We know that $\mathtt{w_1'}(\varphi) = \mathtt{w_2'}(\varphi)$ as $\mathtt{w_1}$ and $\mathtt{w_2}$ differed only on the one variable $v^1$.
\begin{align*}
	\dtv(P(\varphi, \mathtt{w_1}),P(\varphi, \mathtt{w_2})) = \mathtt{w_1'}(\varphi) \times \left(\frac{ \mathtt{w_1}(v^1)}{\mathtt{w_1}(\varphi)} - \frac{\mathtt{w_2}(v^1)}{\mathtt{w_2}(\varphi)}\right)
\end{align*}
All quantities in the above expression are either known constants or they are defined w.r.t the already compiled {\ddnnf}, thus guaranteeing that $\dtv(P(\varphi, \mathtt{w_1}),P(\varphi, \mathtt{w_2}))$ can be computed in poly-time.

\subsection{Extended table of results}
The timeout for all our experiments was set to 7200 seconds. 
\subsubsection{Synthetic PCs}
In the following table, the first column indicates the benchmark, the second and the third indicate the closeness parameter $\varepsilon$ and $\eta$ used in the test.  The fourth column indicates actual $\dtv$ distance between the two benchmark PCs . The fifth column indicates the test outcome and the sixth represents the expected outcome. `A' represents {\accept} and `R' represents {\reject} and `A/R' represents that both `A' and `R' are acceptable outputs. 

\begin{longtable}{p{1.8cm}p{1cm}p{1cm}p{1.8cm}p{1.8cm}p{2.4cm}} \caption{The Extended Table}\\
	\toprule Benchmark&$\varepsilon$&$\eta$&Actual $\dtv$&Result& Expected Result\\ 
	\midrule
14\_4&0.9&0.99&0.773&A&A\\ \midrule 
17\_2&0.75&0.99&0.998&R&R\\ \midrule 
14\_2&0.9&0.99&0.764&A&A\\ \midrule 
18\_3&0.75&0.96&0.930&R&A/R\\ \midrule 
17\_4&0.75&0.99&0.941&R&A/R\\ \midrule 
18\_3&0.8&0.99&0.930&R&A/R\\ \midrule 
17\_1&0.8&0.96&0.874&R&A/R\\ \midrule 
17\_0&0.85&0.94&0.968&A&A/R\\ \midrule 
14\_2&0.85&0.94&0.764&A&A\\ \midrule 
14\_4&0.85&0.94&0.773&A&A\\ \midrule 
15\_4&0.9&0.94&0.941&R&R\\ \midrule 
16\_2&0.9&0.99&0.987&A&A/R\\ \midrule 
14\_0&0.75&0.96&0.771&A&A/R\\ \midrule 
16\_3&0.8&0.9&0.879&A&A/R\\ \midrule 
17\_2&0.75&0.96&0.998&R&R\\ \midrule 
15\_0&0.8&0.9&0.984&R&R\\ \midrule 
18\_0&0.75&0.94&0.994&R&R\\ \midrule 
17\_4&0.75&0.96&0.941&R&A/R\\ \midrule 
18\_4&0.75&0.99&0.907&R&A/R\\ \midrule 
18\_2&0.75&0.99&0.918&R&A/R\\ \midrule 
14\_1&0.75&0.99&0.740&A&A\\ \midrule 
16\_1&0.85&0.9&0.918&R&R\\ \midrule 
17\_0&0.85&0.96&0.968&A&A/R\\ \midrule 
15\_4&0.85&0.94&0.941&R&R\\ \midrule 
17\_0&0.8&0.94&0.968&A&A/R\\ \midrule 
17\_0&0.85&0.99&0.968&A&A/R\\ \midrule 
14\_0&0.9&0.94&0.771&A&A\\ \midrule 
16\_4&0.75&0.96&0.833&A&A/R\\ \midrule 
15\_1&0.85&0.96&0.927&R&A/R\\ \midrule 
14\_4&0.9&0.96&0.773&A&A\\ \midrule 
14\_3&0.8&0.96&0.852&A&A/R\\ \midrule 
14\_2&0.9&0.96&0.764&A&A\\ \midrule 
16\_4&0.75&0.9&0.833&A&A/R\\ \midrule 
17\_3&0.9&0.94&0.914&A&A/R\\ \midrule 
15\_1&0.85&0.9&0.927&R&R\\ \midrule 
16\_4&0.9&0.96&0.833&A&A\\ \midrule 
14\_3&0.8&0.94&0.852&A&A/R\\ \midrule 
16\_1&0.8&0.99&0.918&R&A/R\\ \midrule 
16\_2&0.85&0.96&0.987&A&A/R\\ \midrule 
15\_3&0.75&0.99&0.804&R&A/R\\ \midrule 
14\_3&0.75&0.94&0.852&A&A/R\\ \midrule 
16\_4&0.85&0.96&0.833&A&A\\ \midrule 
18\_0&0.8&0.9&0.994&R&R\\ \midrule 
18\_4&0.8&0.94&0.907&R&A/R\\ \midrule 
18\_4&0.85&0.99&0.907&A&A/R\\ \midrule 
18\_0&0.9&0.99&0.994&R&R\\ \midrule 
15\_1&0.9&0.99&0.927&R&A/R\\ \midrule 
14\_3&0.85&0.96&0.852&A&A/R\\ \midrule 
16\_2&0.75&0.94&0.987&R&R\\ \midrule 
15\_0&0.9&0.96&0.984&R&R\\ \midrule 
18\_4&0.8&0.96&0.907&R&A/R\\ \midrule 
17\_0&0.75&0.9&0.968&R&R\\ \midrule 
18\_1&0.9&0.96&0.993&R&R\\ \midrule 
18\_0&0.9&0.96&0.994&R&R\\ \midrule 
17\_3&0.8&0.99&0.914&A&A/R\\ \midrule 
18\_3&0.85&0.9&0.930&R&R\\ \midrule 
17\_2&0.85&0.94&0.998&R&R\\ \midrule 
15\_1&0.75&0.94&0.927&R&A/R\\ \midrule 
14\_3&0.75&0.9&0.852&R&A/R\\ \midrule 
15\_3&0.8&0.99&0.804&R&A/R\\ \midrule 
17\_3&0.85&0.94&0.914&A&A/R\\ \midrule 
14\_3&0.8&0.9&0.852&R&A/R\\ \midrule 
17\_3&0.75&0.99&0.914&R&A/R\\ \midrule 
14\_3&0.9&0.99&0.852&A&A\\ \midrule 
17\_0&0.75&0.94&0.968&R&R\\ \midrule 
18\_2&0.8&0.99&0.918&R&A/R\\ \midrule 
17\_0&0.8&0.96&0.968&A&A/R\\ \midrule 
17\_1&0.85&0.94&0.874&A&A/R\\ \midrule 
16\_3&0.8&0.94&0.879&A&A/R\\ \midrule 
14\_1&0.85&0.94&0.740&A&A\\ \midrule 
16\_3&0.85&0.99&0.879&A&A/R\\ \midrule 
18\_0&0.85&0.96&0.994&R&R\\ \midrule 
15\_3&0.9&0.94&0.804&A&A\\ \midrule 
16\_4&0.8&0.9&0.833&A&A/R\\ \midrule 
14\_1&0.75&0.96&0.740&A&A\\ \midrule 
16\_2&0.8&0.9&0.987&R&R\\ \midrule 
17\_1&0.75&0.96&0.874&R&A/R\\ \midrule 
15\_1&0.8&0.9&0.927&R&R\\ \midrule 
15\_4&0.8&0.96&0.941&R&A/R\\ \midrule 
18\_2&0.9&0.94&0.918&R&A/R\\ \midrule 
18\_4&0.85&0.94&0.907&R&A/R\\ \midrule 
18\_4&0.85&0.96&0.907&R&A/R\\ \midrule 
16\_3&0.9&0.99&0.879&A&A\\ \midrule 
14\_0&0.75&0.99&0.771&A&A/R\\ \midrule 
16\_0&0.85&0.9&0.954&R&R\\ \midrule 
14\_4&0.85&0.99&0.773&A&A\\ \midrule 
16\_1&0.8&0.9&0.918&R&R\\ \midrule 
17\_1&0.8&0.94&0.874&R&A/R\\ \midrule 
17\_1&0.85&0.99&0.874&A&A/R\\ \midrule 
16\_4&0.9&0.99&0.833&A&A\\ \midrule 
14\_1&0.9&0.94&0.740&A&A\\ \midrule 
17\_0&0.8&0.9&0.968&R&R\\ \midrule 
14\_2&0.75&0.96&0.764&A&A/R\\ \midrule 
15\_0&0.85&0.96&0.984&R&R\\ \midrule 
14\_0&0.8&0.96&0.771&A&A\\ \midrule 
14\_4&0.75&0.96&0.773&A&A/R\\ \midrule 
16\_3&0.75&0.9&0.879&R&A/R\\ \midrule 
17\_2&0.9&0.94&0.998&R&R\\ \midrule 
15\_2&0.85&0.9&0.905&A&A/R\\ \midrule 
14\_4&0.8&0.94&0.773&A&A\\ \midrule 
14\_2&0.8&0.94&0.764&A&A\\ \midrule 
16\_0&0.8&0.99&0.954&R&A/R\\ \midrule 
17\_2&0.85&0.9&0.998&R&R\\ \midrule 
16\_3&0.85&0.96&0.879&A&A/R\\ \midrule 
14\_2&0.75&0.94&0.764&A&A/R\\ \midrule 
18\_1&0.8&0.94&0.993&R&R\\ \midrule 
18\_1&0.85&0.99&0.993&R&R\\ \midrule 
18\_1&0.9&0.99&0.993&R&R\\ \midrule 
16\_4&0.75&0.99&0.833&A&A/R\\ \midrule 
15\_0&0.9&0.99&0.984&R&A/R\\ \midrule 
15\_1&0.9&0.96&0.927&R&A/R\\ \midrule 
16\_0&0.9&0.94&0.954&A&A/R\\ \midrule 
17\_1&0.75&0.9&0.874&R&A/R\\ \midrule 
15\_0&0.8&0.94&0.984&R&R\\ \midrule 
17\_4&0.8&0.99&0.941&R&A/R\\ \midrule 
18\_2&0.85&0.9&0.918&R&R\\ \midrule 
14\_2&0.75&0.9&0.764&A&A/R\\ \midrule 
15\_0&0.8&0.99&0.984&R&A/R\\ \midrule 
14\_2&0.8&0.9&0.764&A&A\\ \midrule 
14\_4&0.8&0.9&0.773&A&A\\ \midrule 
14\_1&0.85&0.9&0.740&A&A\\ \midrule 
17\_0&0.75&0.99&0.968&A&A/R\\ \midrule 
14\_0&0.85&0.9&0.771&A&A\\ \midrule 
17\_1&0.75&0.94&0.874&R&A/R\\ \midrule 
18\_1&0.85&0.94&0.993&R&R\\ \midrule 
18\_1&0.8&0.99&0.993&R&R\\ \midrule 
18\_1&0.75&0.9&0.993&R&R\\ \midrule 
17\_3&0.8&0.96&0.914&R&A/R\\ \midrule 
18\_3&0.9&0.96&0.930&R&A/R\\ \midrule 
16\_2&0.8&0.94&0.987&A&A/R\\ \midrule 
14\_0&0.85&0.94&0.771&A&A\\ \midrule 
16\_2&0.85&0.99&0.987&A&A/R\\ \midrule 
16\_4&0.8&0.94&0.833&A&A/R\\ \midrule 
18\_1&0.85&0.96&0.993&R&R\\ \midrule 
16\_4&0.85&0.99&0.833&A&A\\ \midrule 
15\_2&0.9&0.94&0.905&A&A/R\\ \midrule 
15\_0&0.75&0.9&0.984&R&R\\ \midrule 
16\_0&0.9&0.99&0.954&A&A/R\\ \midrule 
15\_4&0.8&0.9&0.941&R&R\\ \midrule 
17\_0&0.75&0.96&0.968&R&R\\ \midrule 
15\_3&0.8&0.96&0.804&R&A/R\\ \midrule 
18\_3&0.9&0.94&0.930&R&A/R\\ \midrule 
18\_3&0.85&0.94&0.930&R&A/R\\ \midrule 
16\_0&0.8&0.96&0.954&A&A/R\\ \midrule 
17\_4&0.85&0.96&0.941&R&A/R\\ \midrule 
14\_3&0.75&0.99&0.852&A&A/R\\ \midrule 
17\_2&0.85&0.96&0.998&R&R\\ \midrule 
17\_4&0.8&0.94&0.941&R&R\\ \midrule 
16\_2&0.8&0.96&0.987&A&A/R\\ \midrule 
17\_4&0.85&0.99&0.941&R&A/R\\ \midrule 
16\_3&0.8&0.96&0.879&A&A/R\\ \midrule 
17\_1&0.8&0.9&0.874&R&A/R\\ \midrule 
16\_2&0.75&0.96&0.987&R&R\\ \midrule 
15\_3&0.85&0.96&0.804&A&A\\ \midrule 
17\_2&0.8&0.94&0.998&R&R\\ \midrule 
14\_1&0.8&0.96&0.740&A&A\\ \midrule 
14\_0&0.85&0.99&0.771&A&A\\ \midrule 
16\_2&0.75&0.9&0.987&R&R\\ \midrule 
15\_2&0.85&0.94&0.905&A&A/R\\ \midrule 
14\_3&0.85&0.99&0.852&A&A/R\\ \midrule 
15\_3&0.85&0.9&0.804&R&A/R\\ \midrule 
14\_2&0.85&0.99&0.764&A&A\\ \midrule 
16\_3&0.8&0.99&0.879&A&A/R\\ \midrule 
17\_3&0.85&0.9&0.914&A&A/R\\ \midrule 
16\_0&0.85&0.96&0.954&A&A/R\\ \midrule 
14\_1&0.75&0.94&0.740&A&A\\ \midrule 
18\_4&0.8&0.9&0.907&R&R\\ \midrule 
18\_0&0.8&0.94&0.994&R&R\\ \midrule 
14\_3&0.8&0.99&0.852&A&A/R\\ \midrule 
18\_0&0.85&0.99&0.994&R&R\\ \midrule 
18\_2&0.9&0.99&0.918&R&A/R\\ \midrule 
16\_3&0.75&0.99&0.879&A&A/R\\ \midrule 
15\_3&0.9&0.99&0.804&A&A\\ \midrule 
16\_4&0.75&0.94&0.833&A&A/R\\ \midrule 
15\_2&0.9&0.96&0.905&A&A/R\\ \midrule 
16\_1&0.9&0.94&0.918&R&A/R\\ \midrule 
18\_2&0.8&0.96&0.918&R&A/R\\ \midrule 
17\_2&0.9&0.96&0.998&R&R\\ \midrule 
15\_1&0.8&0.94&0.927&R&A/R\\ \midrule 
17\_3&0.9&0.99&0.914&A&A/R\\ \midrule 
15\_1&0.75&0.96&0.927&R&A/R\\ \midrule 
15\_1&0.8&0.99&0.927&R&A/R\\ \midrule 
14\_1&0.8&0.9&0.740&A&A\\ \midrule 
17\_4&0.75&0.94&0.941&R&R\\ \midrule 
18\_0&0.75&0.96&0.994&R&R\\ \midrule 
17\_1&0.75&0.99&0.874&R&A/R\\ \midrule 
17\_2&0.75&0.94&0.998&R&R\\ \midrule 
18\_0&0.8&0.99&0.994&R&R\\ \midrule 
18\_0&0.75&0.9&0.994&R&R\\ \midrule 
17\_2&0.8&0.96&0.998&R&R\\ \midrule 
18\_2&0.9&0.96&0.918&R&A/R\\ \midrule 
16\_1&0.8&0.94&0.918&R&A/R\\ \midrule 
16\_1&0.85&0.99&0.918&R&A/R\\ \midrule 
18\_2&0.85&0.96&0.918&R&A/R\\ \midrule 
15\_1&0.9&0.94&0.927&R&A/R\\ \midrule 
15\_1&0.75&0.9&0.927&R&R\\ \midrule 
16\_1&0.9&0.99&0.918&R&A/R\\ \midrule 
14\_3&0.75&0.96&0.852&A&A/R\\ \midrule 
18\_3&0.75&0.94&0.930&R&A/R\\ \midrule 
15\_2&0.8&0.96&0.905&A&A/R\\ \midrule 
18\_0&0.9&0.94&0.994&R&R\\ \midrule 
18\_1&0.75&0.99&0.993&R&R\\ \midrule 
18\_2&0.85&0.94&0.918&R&A/R\\ \midrule 
17\_3&0.85&0.96&0.914&A&A/R\\ \midrule 
14\_2&0.75&0.99&0.764&A&A/R\\ \midrule 
15\_3&0.85&0.94&0.804&A&A\\ \midrule 
17\_2&0.8&0.9&0.998&R&R\\ \midrule 
16\_3&0.75&0.96&0.879&A&A/R\\ \midrule 
14\_2&0.85&0.9&0.764&A&A\\ \midrule 
15\_2&0.85&0.96&0.905&A&A/R\\ \midrule 
14\_1&0.9&0.96&0.740&A&A\\ \midrule 
16\_1&0.75&0.9&0.918&R&R\\ \midrule 
17\_4&0.9&0.94&0.941&R&R\\ \midrule 
15\_4&0.85&0.9&0.941&R&R\\ \midrule 
16\_4&0.8&0.99&0.833&A&A/R\\ \midrule 
15\_0&0.75&0.99&0.984&R&A/R\\ \midrule 
15\_0&0.85&0.9&0.984&R&R\\ \midrule 
16\_2&0.8&0.99&0.987&A&A/R\\ \midrule 
17\_0&0.85&0.9&0.968&A&A/R\\ \midrule 
16\_1&0.85&0.96&0.918&R&A/R\\ \midrule 
14\_0&0.75&0.94&0.771&A&A/R\\ \midrule 
16\_2&0.9&0.96&0.987&A&A/R\\ \midrule 
18\_3&0.8&0.9&0.930&R&R\\ \midrule 
18\_3&0.8&0.94&0.930&R&A/R\\ \midrule 
14\_2&0.8&0.99&0.764&A&A\\ \midrule 
16\_1&0.9&0.96&0.918&R&A/R\\ \midrule 
18\_3&0.85&0.99&0.930&R&A/R\\ \midrule 
14\_4&0.8&0.99&0.773&A&A\\ \midrule 
16\_0&0.9&0.96&0.954&A&A/R\\ \midrule 
18\_3&0.9&0.99&0.930&R&A/R\\ \midrule 
16\_2&0.75&0.99&0.987&A&A/R\\ \midrule 
14\_0&0.85&0.96&0.771&A&A\\ \midrule 
15\_2&0.9&0.99&0.905&A&A/R\\ \midrule 
16\_1&0.75&0.94&0.918&R&A/R\\ \midrule 
16\_4&0.9&0.94&0.833&A&A\\ \midrule 
15\_3&0.9&0.96&0.804&A&A\\ \midrule 
16\_2&0.9&0.94&0.987&A&A/R\\ \midrule 
18\_3&0.8&0.96&0.930&R&A/R\\ \midrule 
17\_3&0.9&0.96&0.914&A&A/R\\ \midrule 
15\_2&0.8&0.94&0.905&A&A/R\\ \midrule 
17\_0&0.8&0.99&0.968&A&A/R\\ \midrule 
15\_2&0.75&0.94&0.905&R&A/R\\ \midrule 
18\_4&0.85&0.9&0.907&R&R\\ \midrule 
15\_4&0.8&0.99&0.941&R&A/R\\ \midrule 
15\_4&0.75&0.94&0.941&R&R\\ \midrule 
14\_4&0.75&0.9&0.773&A&A/R\\ \midrule 
14\_0&0.8&0.9&0.771&A&A\\ \midrule 
14\_0&0.9&0.99&0.771&A&A\\ \midrule 
18\_1&0.75&0.96&0.993&R&R\\ \midrule 
17\_3&0.75&0.94&0.914&R&A/R\\ \midrule 
18\_3&0.75&0.9&0.930&R&R\\ \midrule 
17\_4&0.85&0.94&0.941&R&R\\ \midrule 
16\_0&0.8&0.94&0.954&A&A/R\\ \midrule 
15\_0&0.85&0.99&0.984&R&A/R\\ \midrule 
16\_0&0.85&0.99&0.954&A&A/R\\ \midrule 
15\_4&0.75&0.9&0.941&R&R\\ \midrule 
15\_1&0.85&0.99&0.927&R&A/R\\ \midrule 
18\_3&0.85&0.96&0.930&R&A/R\\ \midrule 
15\_0&0.9&0.94&0.984&R&R\\ \midrule 
15\_2&0.75&0.9&0.905&R&R\\ \midrule 
15\_2&0.85&0.99&0.905&A&A/R\\ \midrule 
15\_2&0.8&0.9&0.905&A&A/R\\ \midrule 
15\_3&0.85&0.99&0.804&A&A\\ \midrule 
18\_2&0.75&0.94&0.918&R&A/R\\ \midrule 
18\_4&0.75&0.94&0.907&R&A/R\\ \midrule 
15\_1&0.8&0.96&0.927&R&A/R\\ \midrule 
18\_1&0.9&0.94&0.993&R&R\\ \midrule 
18\_0&0.75&0.99&0.994&R&R\\ \midrule 
14\_3&0.85&0.9&0.852&A&A/R\\ \midrule 
16\_3&0.85&0.9&0.879&A&A/R\\ \midrule 
16\_1&0.8&0.96&0.918&R&A/R\\ \midrule 
14\_1&0.85&0.99&0.740&A&A\\ \midrule 
15\_0&0.85&0.94&0.984&R&R\\ \midrule 
17\_2&0.85&0.99&0.998&R&R\\ \midrule 
14\_2&0.9&0.94&0.764&A&A\\ \midrule 
14\_4&0.9&0.94&0.773&A&A\\ \midrule 
17\_3&0.8&0.9&0.914&R&R\\ \midrule 
16\_0&0.75&0.96&0.954&R&A/R\\ \midrule 
14\_0&0.9&0.96&0.771&A&A\\ \midrule 
17\_1&0.9&0.94&0.874&A&A\\ \midrule 
16\_0&0.75&0.9&0.954&R&R\\ \midrule 
14\_1&0.8&0.94&0.740&A&A\\ \midrule 
15\_1&0.75&0.99&0.927&R&A/R\\ \midrule 
17\_1&0.85&0.9&0.874&R&A/R\\ \midrule 
18\_2&0.8&0.9&0.918&R&R\\ \midrule 
18\_2&0.8&0.94&0.918&R&A/R\\ \midrule 
14\_1&0.8&0.99&0.740&A&A\\ \midrule 
18\_2&0.85&0.99&0.918&R&A/R\\ \midrule 
18\_4&0.9&0.99&0.907&A&A/R\\ \midrule 
16\_1&0.75&0.99&0.918&R&A/R\\ \midrule 
14\_1&0.85&0.96&0.740&A&A\\ \midrule 
16\_0&0.75&0.94&0.954&R&R\\ \midrule 
15\_4&0.9&0.96&0.941&R&A/R\\ \midrule 
17\_2&0.75&0.9&0.998&R&R\\ \midrule 
16\_3&0.9&0.94&0.879&A&A\\ \midrule 
18\_0&0.8&0.96&0.994&R&R\\ \midrule 
17\_4&0.75&0.9&0.941&R&R\\ \midrule 
15\_3&0.8&0.94&0.804&R&A/R\\ \midrule 
17\_1&0.8&0.99&0.874&R&A/R\\ \midrule 
18\_1&0.85&0.9&0.993&R&R\\ \midrule 
15\_3&0.75&0.94&0.804&R&A/R\\ \midrule 
14\_1&0.75&0.9&0.740&A&A\\ \midrule 
17\_1&0.9&0.99&0.874&A&A\\ \midrule 
15\_3&0.75&0.96&0.804&R&A/R\\ \midrule 
18\_4&0.75&0.96&0.907&R&A/R\\ \midrule 
14\_1&0.9&0.99&0.740&A&A\\ \midrule 
18\_2&0.75&0.96&0.918&R&A/R\\ \midrule 
18\_4&0.8&0.99&0.907&R&A/R\\ \midrule 
18\_4&0.75&0.9&0.907&R&R\\ \midrule 
18\_2&0.75&0.9&0.918&R&R\\ \midrule 
17\_4&0.8&0.96&0.941&R&A/R\\ \midrule 
14\_3&0.85&0.94&0.852&A&A/R\\ \midrule 
18\_4&0.9&0.96&0.907&A&A/R\\ \midrule 
17\_3&0.75&0.9&0.914&R&R\\ \midrule 
17\_4&0.9&0.96&0.941&R&A/R\\ \midrule 
15\_3&0.75&0.9&0.804&R&A/R\\ \midrule 
16\_0&0.8&0.9&0.954&R&R\\ \midrule 
17\_3&0.75&0.96&0.914&R&A/R\\ \midrule 
15\_3&0.8&0.9&0.804&R&A/R\\ \midrule 
18\_1&0.75&0.94&0.993&R&R\\ \midrule 
16\_1&0.85&0.94&0.918&R&A/R\\ \midrule 
16\_3&0.85&0.94&0.879&A&A/R\\ \midrule 
18\_4&0.9&0.94&0.907&A&A/R\\ \midrule 
15\_0&0.8&0.96&0.984&R&R\\ \midrule 
16\_0&0.85&0.94&0.954&A&A/R\\ \midrule 
14\_4&0.85&0.9&0.773&A&A\\ \midrule 
18\_3&0.75&0.99&0.930&R&A/R\\ \midrule 
17\_0&0.9&0.96&0.968&A&A/R\\ \midrule 
18\_0&0.85&0.94&0.994&R&R\\ \midrule 
17\_1&0.9&0.96&0.874&A&A\\ \midrule 
16\_4&0.8&0.96&0.833&A&A/R\\ \midrule 
16\_2&0.85&0.9&0.987&A&A/R\\ \midrule 
17\_1&0.85&0.96&0.874&A&A/R\\ \midrule 
14\_4&0.75&0.99&0.773&A&A/R\\ \midrule 
16\_4&0.85&0.9&0.833&A&A\\ \midrule 
17\_3&0.8&0.94&0.914&R&A/R\\ \midrule 
15\_1&0.85&0.94&0.927&R&A/R\\ \midrule 
17\_3&0.85&0.99&0.914&A&A/R\\ \midrule 
14\_3&0.9&0.94&0.852&A&A\\ \midrule 
17\_4&0.8&0.9&0.941&R&R\\ \midrule 
16\_1&0.75&0.96&0.918&R&A/R\\ \midrule 
15\_4&0.85&0.96&0.941&R&A/R\\ \midrule 
14\_2&0.8&0.96&0.764&A&A\\ \midrule 
14\_3&0.9&0.96&0.852&A&A\\ \midrule 
17\_0&0.9&0.94&0.968&A&A/R\\ \midrule 
14\_4&0.8&0.96&0.773&A&A\\ \midrule 
16\_3&0.9&0.96&0.879&A&A\\ \midrule 
15\_4&0.75&0.99&0.941&R&A/R\\ \midrule 
14\_0&0.8&0.94&0.771&A&A\\ \midrule 
17\_4&0.85&0.9&0.941&R&R\\ \midrule 
15\_2&0.75&0.99&0.905&A&A/R\\ \midrule 
14\_4&0.75&0.94&0.773&A&A/R\\ \midrule 
17\_4&0.9&0.99&0.941&A&A/R\\ \midrule 
18\_1&0.8&0.9&0.993&R&R\\ \midrule 
15\_4&0.85&0.99&0.941&R&A/R\\ \midrule 
14\_0&0.8&0.99&0.771&A&A\\ \midrule 
17\_2&0.9&0.99&0.998&R&R\\ \midrule 
14\_0&0.75&0.9&0.771&A&A/R\\ \midrule 
14\_4&0.85&0.96&0.773&A&A\\ \midrule 
16\_0&0.75&0.99&0.954&R&A/R\\ \midrule 
14\_2&0.85&0.96&0.764&A&A\\ \midrule 
15\_4&0.9&0.99&0.941&A&A/R\\ \midrule 
16\_3&0.75&0.94&0.879&R&A/R\\ \midrule 
16\_4&0.85&0.94&0.833&A&A\\ \midrule 
18\_1&0.8&0.96&0.993&R&R\\ \midrule 
15\_0&0.75&0.96&0.984&R&R\\ \midrule 
16\_2&0.85&0.94&0.987&A&A/R\\ \midrule 
15\_4&0.8&0.94&0.941&R&R\\ \midrule 
17\_2&0.8&0.99&0.998&R&R\\ \midrule 
18\_0&0.85&0.9&0.994&R&R\\ \midrule 
15\_0&0.75&0.94&0.984&R&R\\ \midrule 
15\_4&0.75&0.96&0.941&R&A/R\\ \midrule 
17\_0&0.9&0.99&0.968&A&A/R\\ \midrule 
15\_2&0.8&0.99&0.905&A&A/R\\ \midrule 
15\_2&0.75&0.96&0.905&R&A/R\\  
		\bottomrule\end{longtable}\raggedbottom

\subsubsection{Real-world PCs} 
In the following table, the first column indicates the benchmark, the second indicates the time required for the test, and the third column indicates the test outcome. `A' represents {\accept} and `R' represents {\reject}. 

\begin{longtable} {p{4.8cm}p{2.6cm}p{2.5cm}} \caption{The Extended Table}\\
	 \toprule Benchmark&$\teq$(s)&Result\\ 
	 \midrule 
	 or-70-10-8-UC-10\_0&23.2&A\\  
or-70-10-8-UC-10\_1&22.72&R\\  
or-70-10-8-UC-10\_2&22.92&R\\  
or-70-10-8-UC-10\_3&22.87&R\\  
or-70-10-8-UC-10\_4&22.78&R\\  
or-70-10-8-UC-10\_5&23.06&R\\  
or-70-10-8-UC-10\_6&22.99&R\\  
or-70-10-8-UC-10\_7&22.93&R\\  
or-70-10-8-UC-10\_8&22.82&R\\  
or-70-10-8-UC-10\_9&22.82&R\\  
s641\_15\_7\_0&33.66&A\\  
s641\_15\_7\_1&33.4&R\\  
s641\_15\_7\_2&33.45&R\\  
s641\_15\_7\_3&33.32&R\\  
s641\_15\_7\_4&33.51&R\\  
s641\_15\_7\_5&33.21&R\\  
s641\_15\_7\_6&33.46&R\\  
s641\_15\_7\_7&33.23&R\\  
s641\_15\_7\_8&33.61&R\\  
s641\_15\_7\_9&33.51&R\\  
or-50-5-4\_0&414.17&A\\  
or-50-5-4\_1&414.84&R\\  
or-50-5-4\_2&410.16&R\\  
or-50-5-4\_3&414.15&R\\  
or-50-5-4\_4&410.07&R\\  
or-50-5-4\_5&412.27&R\\  
or-50-5-4\_6&414.77&R\\  
or-50-5-4\_7&415.19&R\\  
or-50-5-4\_8&416.84&R\\  
or-50-5-4\_9&408.59&R\\  
ProjectService3.sk\_12\_55\_0&356.58&A\\  
ProjectService3.sk\_12\_55\_1&353.77&R\\  
ProjectService3.sk\_12\_55\_2&355.93&R\\  
ProjectService3.sk\_12\_55\_3&356.11&R\\  
ProjectService3.sk\_12\_55\_4&356.15&A\\  
ProjectService3.sk\_12\_55\_5&355.64&R\\  
ProjectService3.sk\_12\_55\_6&357.89&R\\  
ProjectService3.sk\_12\_55\_7&356.69&R\\  
ProjectService3.sk\_12\_55\_8&353.36&R\\  
ProjectService3.sk\_12\_55\_9&356.14&R\\  
s713\_15\_7\_0&24.56&R\\  
s713\_15\_7\_1&24.68&R\\  
s713\_15\_7\_2&24.28&R\\  
s713\_15\_7\_3&24.47&R\\  
s713\_15\_7\_4&24.65&R\\  
s713\_15\_7\_5&24.32&R\\  
s713\_15\_7\_6&24.4&R\\  
s713\_15\_7\_7&24.39&R\\  
s713\_15\_7\_8&24.86&A\\  
s713\_15\_7\_9&24.41&R\\  
or-100-10-2-UC-30\_0&31.11&R\\  
or-100-10-2-UC-30\_1&31.16&R\\  
or-100-10-2-UC-30\_2&31.04&R\\  
or-100-10-2-UC-30\_3&31.13&R\\  
or-100-10-2-UC-30\_4&31.14&R\\  
or-100-10-2-UC-30\_5&31.04&A\\  
or-100-10-2-UC-30\_6&31.03&R\\  
or-100-10-2-UC-30\_7&31.13&R\\  
or-100-10-2-UC-30\_8&31.17&R\\  
or-100-10-2-UC-30\_9&31.0&R\\  
s1423a\_3\_2\_0&153.8&R\\  
s1423a\_3\_2\_1&152.37&R\\  
s1423a\_3\_2\_2&152.01&R\\  
s1423a\_3\_2\_3&150.96&R\\  
s1423a\_3\_2\_4&152.64&R\\  
s1423a\_3\_2\_5&153.13&A\\  
s1423a\_3\_2\_6&151.52&R\\  
s1423a\_3\_2\_7&152.53&R\\  
s1423a\_3\_2\_8&152.4&R\\  
s1423a\_3\_2\_9&152.81&R\\  
s1423a\_7\_4\_0&104.28&R\\  
s1423a\_7\_4\_1&103.4&R\\  
s1423a\_7\_4\_2&103.82&R\\  
s1423a\_7\_4\_3&104.18&R\\  
s1423a\_7\_4\_4&103.95&R\\  
s1423a\_7\_4\_5&103.59&R\\  
s1423a\_7\_4\_6&104.31&R\\  
s1423a\_7\_4\_7&104.93&R\\  
s1423a\_7\_4\_8&104.93&A\\  
s1423a\_7\_4\_9&103.51&R\\  
or-50-5-10\_0&282.09&R\\  
or-50-5-10\_1&282.49&R\\  
or-50-5-10\_2&279.63&R\\  
or-50-5-10\_3&281.8&R\\  
or-50-5-10\_4&280.69&R\\  
or-50-5-10\_5&279.91&R\\  
or-50-5-10\_6&283.05&A\\  
or-50-5-10\_7&282.69&R\\  
or-50-5-10\_8&279.65&R\\  
or-50-5-10\_9&282.97&R\\  
or-60-20-6-UC-20\_0&359.89&R\\  
or-60-20-6-UC-20\_1&362.3&R\\  
or-60-20-6-UC-20\_2&363.1&R\\  
or-60-20-6-UC-20\_3&363.11&R\\  
or-60-20-6-UC-20\_4&362.76&R\\  
or-60-20-6-UC-20\_5&358.76&R\\  
or-60-20-6-UC-20\_6&363.32&A\\  
or-60-20-6-UC-20\_7&358.41&R\\  
or-60-20-6-UC-20\_8&358.8&R\\  
or-60-20-6-UC-20\_9&362.8&R\\  
\bottomrule\end{longtable}

\end{document}